\documentclass[a4paper,12pt]{article}

 \usepackage[T1]{fontenc}         
 \usepackage[utf8]{inputenc}
 \usepackage{lmodern}             
 \usepackage{amsmath}
 \usepackage{bbm}
 \usepackage{amssymb}
 \usepackage{amsthm}
 \usepackage{graphicx}
 \usepackage{vmargin}
 \usepackage{color}
 \usepackage{epsfig}
 \usepackage[only,llbracket,rrbracket]{stmaryrd}
 \usepackage[dvipdfm]{hyperref}

\newtheorem{remark}{Remark}[section]

\newtheorem{thm}{Theorem}
\newtheorem{prop}{Proposition}[section]
\newtheorem{lemma}{Lemma}

\makeatletter
\renewcommand\section{\setcounter{equation}{0}\@startsection {section}{1}{\z@}%
                                   {-3.5ex \@plus -1ex \@minus -.2ex}%
                                   {2.3ex \@plus.2ex}%
                                   {\normalfont\Large\bfseries}}
\renewcommand\theequation
  {\ifnum \c@section>\z@ \thesection.\fi \@arabic\c@equation}
\makeatother

\title{A continuum model for alignment of self-propelled particles with anisotropy and density-dependent parameters}
\author{Amic Frouvelle\thanks{Institut de Mathématiques de Toulouse, CNRS -- UMR 5219, Université de Toulouse, F-31062 Toulouse, France, \href{mailto:amic.frouvelle@math.univ-toulouse.fr}{amic.frouvelle@math.univ-toulouse.fr}}}
\date{}

\begin{document}
\maketitle

\begin{abstract}
We consider the macroscopic model derived by Degond and Motsch from a time-continuous version of the Vicsek model, describing the interaction orientation in a large number of self-propelled particles. In this article, we study the influence of a slight modification at the individual level, letting the relaxation parameter depend on the local density and taking in account some anisotropy in the observation kernel (which can model an angle of vision). 

The main result is a certain robustness of this macroscopic limit and of the methodology used to derive it. With some adaptations to the concept of generalized collisional invariants, we are able to derive the same system of partial differential equations, the only difference being in the definition of the coefficients, which depend on the density. This new feature may lead to the loss of hyperbolicity in some regimes.

We provide then a general method which enables us to get asymptotic expansions of these coefficients. These expansions shows, in some effective situations, that the system is not hyperbolic. This asymptotic study is also useful to measure the influence of the angle of vision in the final macroscopic model, when the noise is small. 
\end{abstract}



\newpage

\section{Introduction}

The study of complex particle systems has given rise to some challenging issues~\cite{bellomo2009mathematics}, in a mathematical point of view. 
One of the interesting problem is to understand how a collective behavior can emerge with only localized interactions.

The Vicsek model~\cite{vicsek1995novel} has been proposed as a minimalist model describing the behavior of individuals inside animal societies such as fish schools or flocks of birds. 
It is a minimal version of a more complete and realistic model~\cite{aoki1982simulation,reynolds1987flocks,huth1992simulation,couzin2002collective} based on three zones (of repulsion, alignment, and attraction). 
The Vicsek model only considers the alignment behavior, getting around the problem of confining the particles in the same region by imposing spatial periodicity (particles move on the flat torus). 
All the particles have constant speed and synchronously update their direction according to their neighbors, their new orientation vector being given by the mean direction (subjected to some angular noise) of all particles at distance less than a given radius. 
As the noise decreases (or the density increases), one can observe a phenomenon of phase transition, from a regime of disordered particles, to an ordered phase with strong correlations between orientations of particles~\cite{vicsek1995novel,aldana2003phase,gregoire2004onset}.

Two main difficulties arise when we try to derive a macroscopic limit of this individual based model. 
First of all, the system is discrete in time: the time step is fixed, and the model is not built in the goal of letting it tend to zero. 
The second problem is that, except for the total mass, there is no obvious conservation relation, so a good candidate for a macroscopic model would probably be a non-conservative system of partial differential equations, but we lack conservation relations to obtain any equation other than the conservation of mass.

In~\cite{degond2008continuum}, Degond and Motsch have proposed an approach to handle these two complications. 
First, they provide a time-continuous version of the individual based model, introducing a rate of relaxation towards the local mean direction, under the form of a new parameter~$\nu$, which can be viewed as a frequency of interaction between a particle and its neighbors. 
It is therefore possible to derive a kinetic mean-field limit of this model.
Then they develop a method, defining the notion of generalized collisional invariants, which allows to derive the formal limit of this kinetic mean-field model, at large scale in space and time. 
This continuum limit is a non-conservative system of PDE for the local mass and the local orientation. 
Moreover, this system is proved to be hyperbolic.

The goal of the present paper is to confirm the ability of this type of macroscopic model to describe the large scale dynamics of systems of self-propelled particles with orientation alignment, and to show that the notion of generalized collisional invariants is well adapted to derive this model from the microscopic mechanism of alignment. 
This was shown in~\cite{degond2010macroscopic} for a different type of alignment, based on the curvature control (for a model of displacement introduced in~\cite{degond2008large}, designed to fit biological experiments~\cite{gautrais2009analyzing}): the method which uses the generalized collisional invariants is successful to derive a macroscopic model which is the same as the “Vicsek hydrodynamics” of~\cite{degond2008continuum}, but for the definition of the coefficients in the model. 
We will show that this is also the case when we slightly modify the individual model, in order to be more coherent with some numerical observations, and to model the influence of an angle of vision.

One of the properties of the macroscopic model of~\cite{degond2008continuum} which fails to represent the numerical observations is that the local equilibria have a constant order parameter. 
This order parameter is indeed only related to the ratio between the frequency~$\nu$ of interaction and the intensity~$d$ of the noise. 
In numerical experiments on the Vicsek model~\cite{chate2008collective}, the order parameter depends on the local density of particles: one can observe, at large time, formation of travelling bands of high density, strongly ordered, moving through a disordered area of low density.

The first refinement on the model will be to define the local density~$\bar{\rho}$ in the particular model as the mean number of particles in a neighboring area, and to make the parameters~$\nu$ and~$d$ depend on~$\bar{\rho}$. 
In a modelling point of view, this could be interpreted as the fact that, due to social pressure for example, it is more likely to move and update the direction when there is a large number of particles around (so~$\nu$ increases with~$\bar{\rho}$), and that the fluctuation in the estimation of the mean velocity is smaller when a lot of particles are taken in account in the neighborhood (so~$d$ decreases with~$\bar{\rho}$, see for example the way the vectorial noise is defined in~\cite{gregoire2004onset,chate2008collective}). This dependence on local parameters for the noise parameter has been introduced in other models of collective behavior~\cite{duan2010kinetic,yates2009inherent}.

In~\cite{degond2008continuum}, the parameter~$\nu$ does also depend on the angle between one particle direction and its target direction. 
For convenience, we will not take this in account for most results but, in some simple cases, computations have been done and will be given in appendix.

Another property of the mean-field model of~\cite{degond2008continuum} is that the type of local equilibrium for the rescaled model is unique. 
The loss of this uniqueness could play a role to understand the formation of patterns such as the travelling bands~\cite{fan2004pattern}, and in some models of rod-like particles, we have indeed bistability in a certain regime~\cite{liu2005axial}. 
This is not the case here, and we will see in Section~\ref{elements} that we still have a unique kind of equilibrium associated to a local density~$\rho$ and a local orientation~$\Omega$.

The second refinement is to take into account the influence of an “angle of vision” in the model. 
In the original Vicsek model, the target orientation for a given particle is chosen to be the mean orientation of the neighbors located in a ball centered on this particle. 
We will use here a more general kernel of observation which can be non-isotropic. This refinement has been proposed in various models of swarming~\cite{agueh2011analysis,li2008minimal}.

The main result of this paper is that the formal macroscopic limit of this model take the same form as the previous “Vicsek Hydrodynamic model” of~\cite{degond2008continuum}, consisting in a conservation equation for the local density of particles, and an evolution equation for the mean orientation, which is not conservative (the velocity is constrained to be on the unit sphere). 
This system of PDE is quite similar to the Euler equations for gas dynamics, but presents some specific issues, for example there are two different velocities of propagation.
The difference between our model and the macroscopic model of~\cite{degond2008continuum} relies on the definition of the coefficients of this model, and on the fact that they depend on the local density. This last feature allows the model to lose the property of hyperbolicity. 

In Section~\ref{presentation}, we present the individual model and the final macroscopic model, focusing on how the two refinements are taken in account, and what are the consequences at the macroscopic level.

In Section~\ref{elements}, we provide elements of the derivation of this macroscopic model, following the method of~\cite{degond2008continuum}, but emphasizing the details which are specific to this study. 
We also give the method in a general~$n$-dimensional framework (previously the method was only done in three dimensions). 
The case of the dimension~$2$ is special, since we are able to give an explicit expression of the coefficients.

In Section~\ref{properties}, we study the properties of the macroscopic model. We prove that when one of the coefficients is negative, then the system is not hyperbolic. 
We describe the region of hyperbolicity for a system which depends only on one space variable, and we discuss the influence of “angle of vision”. 

Finally, in Section~\ref{asymptotics}, we provide a general method which gives an asymptotic expansion of the coefficients in any dimension, in the limit of a small or a large concentration parameter. With these expansions, we are able to study the qualitative influence of the “angle of vision” in the final macroscopic model, and to give examples for which the hyperbolicity is indeed lost.

\section[Individual and continuum dynamics of a modified Vicsek model]{Individual and continuum dynamics of a modified Vicsek model}
\label{presentation}
We start by presenting the individual-based model and the continuum model we obtain in the limit of a large number of particles, when observed at large scale, in space and time. Elements of the derivation of this macroscopic model will be given in section~\ref{elements}.

\subsection{Starting point: particle dynamics}

Here, we briefly recall the time-continuous version of the Vicsek model, and introduce how we take into account the anisotropy of observation and the dependence on the local density for the rate of relaxation and the intensity of the noise. 

We consider a system of~$N$ particles with positions~$X_k$ in~$\mathbb{R}^n$ (with~$k\in\llbracket1,N\rrbracket$) and orientations~$\omega_k$ in the unit sphere~$\mathbb{S}_{n-1}$, which we will simply write~$\mathbb{S}$.

For each particle, we first define a local mean orientation~$\bar{\omega}_k$ (considered as a target direction) and a local density~$\bar{\rho}_k$. 
In the original model of Vicsek~\cite{vicsek1995novel}, the mean orientation~$\bar{\omega}_k$ is computed on all the neighbors within a given radius~$R$. 
Here we take the mean according to a kernel of observation~$K$, which can be more general than the indicator function of the ball of radius~$R$, as in the time-continuous version of~\cite{degond2008continuum}. 
The kernel therein depends only on the distance between the given particle and a given neighbor, so the refinement here is that it can also depend on (the cosine of) the angle between the orientation of the first particle and the right line joining the two particles:
\begin{equation}
\bar{\omega}_k = \frac{\bar J_k}{|\bar J_k|}, \text{ where } \bar J_k = 
\frac1N\sum_{j=1}^N K \left(|X_j - X_k|, \tfrac{X_j - X_k}{|X_j - X_k|}\cdot\omega_k\right) \omega_j . 
\label{omega_k_bar} 
\end{equation}

For example, to take into account only the neighbors located “in front”, and within a given radius~$R$, of one particle, the kernel would be~$K(r,\gamma)=\mathbbm{1}_{\{r\leqslant R\}}\mathbbm{1}_{\{\gamma\geqslant0\}}$.
We proceed in an analogous way to compute the local density~$\bar{\rho}_k$, which may use another kernel~$\widetilde{K}$:
\begin{gather}
\label{rho_k_bar} 
\bar{\rho}_k = \frac1N\sum_{j=1}^N \widetilde K \left(|X_j - X_k|, \tfrac{X_j - X_k}{|X_j - X_k|}\cdot\omega_k\right).
\end{gather}

We now turn to the dynamics of the particle system.
The~$k^\text{th}$ particle moves at constant speed~$1$, following its orientation~$\omega_k$. 
This last one relax towards the mean orientation~$\bar{\omega}_k$ of its neighbors, with rate~$\nu$ (depending in the local mean density~$\bar{\rho}_k$), under the constraint that~$\omega_k$ is of norm~$1$.
Finally, this orientation~$\omega_k$ is subjected to a Brownian motion (see~\cite{hsu2002stochastic} for more details on how to define such an object on a Riemannian manifold, such as the unit sphere here) of intensity~$d$, which will also depend on the density~$\bar{\rho}_k$.
The model takes then the form of~$2N$ coupled stochastic differential equations, which have to be understood in the Stratonovich sense:
\begin{align}
\mathrm d X_k &= \omega_k \mathrm d t\label{SDE_pos}, \\
\mathrm d \omega_k &= \nu(\bar{\rho}_k) (\mathrm{Id} - \omega_k \otimes \omega_k) \bar{\omega}_k \, \mathrm d t + \sqrt{2d(\bar{\rho}_k)} \, (\mathrm{Id} - \omega_k \otimes \omega_k)\circ\mathrm d B^k_t , 
\label{SDE_orient} 
\end{align}
where~$(B^k_t)$ are independent standard Brownian motions on~$\mathbb{R}^n$.

Here we denote by~$\mathrm{Id} - \omega_k \otimes \omega_k$ the projection on the plane orthogonal to~$\omega_k$, that is to say~$(\mathrm{Id} - \omega \otimes \omega)\upsilon=\upsilon - (\upsilon\cdot\omega)\omega$. This projection is necessary to keep~$\omega_k$ on the unit sphere. 
The term~$(\mathrm{Id} - \omega_k \otimes \omega_k) \bar{\omega}_k$ can also be written~$\nabla_\omega(\omega\cdot\bar{\omega}_k)|_{\omega=\omega_k}$, where~$\nabla_\omega$ is the tangential gradient on the unit sphere.
The deterministic part of the SDE~\eqref{SDE_orient} can then be written~$\frac{\mathrm d \omega_k}{\mathrm d t} = \nu(\bar{\rho}_k)\nabla_\omega(\omega\cdot\bar{\omega}_k)|_{\omega=\omega_k}$, which is indeed a relaxation towards~$\bar{\omega}_k$ (where the function~$\omega\mapsto\omega\cdot\bar{\omega}_k$ reaches its maximum), with rate~$\nu(\bar{\rho}_k)$.

Since the local density~$\bar{\rho}$ only appears through the functions~$\nu$ and~$d$, we can assume the following normalization for the kernel~$\widetilde{K}$:
\begin{equation}
\int_{ \xi \in \mathbb{R}^n} \widetilde K (|\xi |,\tfrac{\xi}{|\xi |}\cdot\omega) \mathrm d\xi=1\,. 
\label{rho_bar_normalization}
\end{equation}
This normalization condition (which does not depend on~$\omega\in\mathbb{S}$) means that the density is chosen to be~$1$ in the limit of a uniform distribution of the~$N$ particles in a region of unit volume.
This is not necessary to take a similar condition for the kernel~$K$, since~$\bar{\omega}_k$, defined at equation~$\eqref{omega_k_bar}$, is independent of such a normalization.


In~\cite{degond2008continuum}, the relaxation coefficient~$\nu$ depends on (the cosine of) the angle between the orientation of one particle and the target direction, in order to take into account the “ability to turn”. 
This would amount to replace~$\nu(\bar{\rho}_k)$ by~$\nu(\bar{\rho}_k,\omega_k\cdot\bar{\omega}_k)$ in~\eqref{SDE_orient}. 
With our new features here, this would involve many more computations, but, following exactly the same method, this leads to the same conclusion. 
For simplicity here, we will work without this dependence. 
We will only present the final results with this dependence in some special cases, and add remarks to explain the difference in the steps of the derivation of the macroscopic model.

Numerical simulations tend to show that this time-continuous individual based model present the same behavior at large scale as the discrete one (for example the formation of bands, as in~\cite{chate2008collective}), in the case where~$\nu$ and~$d$ are constant, and the observation kernel is isotropic, as in~\cite{degond2008continuum}. 
We can expect to observe the same behavior, even when~$\nu$ and~$d$ depend on the local density~$\bar{\rho}$.
More precise investigations on the numerical comparison between the original discrete and the present time-continuous dynamical systems are in progress.

\subsection{The continuum model}
In this paper, following the approach of~\cite{degond2008continuum}, we derive, from the particle dynamics~\eqref{SDE_pos}-\eqref{SDE_orient} introduced in the previous subsection, the following continuum model, the functions~$\rho(x,t)>0$ and~$\Omega(x,t)\in\mathbb{S}$ describing the average density and particle direction at a given point~$x\in \mathbb{R}^n$:
\begin{gather}
\partial_t \rho + \nabla_x \cdot (c_1(\rho) \rho \Omega) = 0,
\label{mass_eq} \\
\rho \, \left( \partial_t \Omega + c_2(\rho) (\Omega \cdot \nabla_x) \Omega \right) + \lambda(\rho) \, (\mathrm{Id} - \Omega \otimes \Omega) \nabla_x \rho = 0,
\label{Omega_eq}
\end{gather}
where the functions~$c_1$,~$c_2$, and~$\lambda$ will be specified later on: see~\eqref{def_c1} and~\eqref{def_c2}-\eqref{def_lambda}.

This system of first order partial differential equations shows similarities with the Euler system of isothermal compressible gases, but also some important differences.

Equation~\eqref{mass_eq} is the conservation of mass: the density~$\rho$ moves through direction~$\Omega$ with velocity~$c_1(\rho)$. 
We will see that this velocity, taking values between~$0$ and~$1$, plays the role of an order parameter: when the directions of the particles are strongly correlated (close to~$\Omega$), density moves with velocity close to~$1$. This order parameter depends only on the ratio between~$\nu(\rho)$ (the alignment strength) and~$d(\rho)$ (the noise intensity). 

Equation~\eqref{Omega_eq} describes the evolution of the direction~$\Omega$, the norm of which is constrained to be constant (the projection~$\mathrm{Id}-\Omega\otimes\Omega$ insures that the dynamics take place on the hyperplane orthogonal to~$\Omega$). This constraint implies that the equation is not conservative which is the counterpart of the fact that, at the microscopic level, the only conservative quantity is the mass.
The perturbations of this vector travel with velocity~$c_2(\rho)$, influenced by a term playing the role of pressure due to the density, of intensity~$\lambda(\rho)$. It is important to see that in general (and contrary to the classical Euler system), the two convection speeds~$c_1$ and~$c_2$ are different, which means that the perturbations on the mean orientation do not travel at the same velocity as the “fluid”.

This macroscopic model is the same as the “Vicsek hydrodynamics” of~\cite{degond2008continuum}, except for the definitions of these speeds~$c_1$ and~$c_2$, and of the parameter~$\lambda$. 
This confirms the ability of this model to describe the global dynamics of systems of self-propelled particles with constant speed and alignment interactions (this was also the result of~\cite{degond2010macroscopic}). 

However, the parameters depend here on the density~$\rho$, and their expressions are slightly different (due to this dependence and to the anisotropy of the kernel of observation).
This leads to strong differences in their behavior, as we will see in Section~\ref{properties}, devoted to the investigation of the properties of~\eqref{mass_eq}-\eqref{Omega_eq}. 
For example the parameter~$\lambda$ can be negative, which implies the loss of hyperbolicity. And because of the non isotropy of the observation kernel, the convection speed~$c_2$ can take a large range of values, from negative if the kernel is strongly directed forward, to higher than~$c_1$, if particles are more influenced by neighbors behind them than those in front of them (this has been observed in locust migratory bands~\cite{bazazi2008collective}, where the individuals have “cannibalistic interactions” and avoid to be eaten by those approaching from behind).

\section{Elements of the derivation of the continuum model}
\label{elements}

This derivation proceeds like in~\cite{degond2008continuum} but there are significant differences due to the additional complexity. In this section, we briefly recall the method of Degond and Motsch and will focus on the points which are specific to the present study.

The derivation proceeds in several steps. The first one consists in writing a kinetic version of the particle dynamics.

\subsection{Step 1: mean-field model}
Let~$f(x,\omega,t)$ be the probability density of finding one particle at position~$x\in\mathbb{R}^n$, orientation~$\omega\in\mathbb{S}$ and time~$t\geqslant0$. The mean-field version of~\eqref{SDE_pos}-\eqref{SDE_orient} is given by 
\begin{equation}
\partial_t f + \omega \cdot \nabla_x f + \nabla_\omega \cdot (F f) = \nabla_\omega \cdot (\sqrt{d(\bar{\rho})}\nabla_\omega\sqrt{d(\bar{\rho})} f) , 
\label{KFP} 
\end{equation}
with
\begin{align*}
F (x, \omega, t) &= \nu(\bar{\rho}) \, (\mathrm{Id} - \omega \otimes \omega) \bar{\omega} (x,\omega,t) ,\\
\bar{\rho} (x, \omega, t) &= \int_{ y \in \mathbb{R}^n , \, \upsilon \in \mathbb{S} } \widetilde K\left(|y-x|,\tfrac{y-x}{|y-x|}\cdot\omega\right) \, f(y, \upsilon,t) \, \mathrm d y \, \mathrm d\upsilon\, ,\\
\bar{\omega} (x, \omega, t) &= \frac{J(x,\omega,t)}{|J(x,\omega,t)|}, \\
J(x,\omega,t) &= \int_{ y \in \mathbb{R}^n , \, \upsilon \in \mathbb{S} } K\left(|y-x|,\tfrac{y-x}{|y-x|}\cdot\omega\right) \, 
\upsilon \, f(y, \upsilon,t) \, \mathrm d y \, \mathrm d\upsilon \, .
\end{align*}
The first equation~\eqref{KFP} is the so called Kolmogorov--Fokker--Planck equation. The force term~$F(x,\omega,t)$ corresponds to the orientation interaction.

First of all, if there is no noise (that is~$d(\bar{\rho})=0$ in equations~\eqref{SDE_pos}-\eqref{SDE_orient}, which become ordinary differential equations), the formal derivation of this system is easy: the usual methodology shows that the empirical distribution (see~\cite{spohn1991large}) satisfies the equation~\eqref{KFP}, with~$d=0$.

Difficulties appear when the noise is added. A method consisting in writing the BBGKY hierarchy (see~\cite{ha2008particle}, applied to the Cucker-Smale model of self-propelled particles) would not in that case reduce to a evolution equation involving only the one-particle and the two-particles distributions, since the interaction is not a sum of binary interactions.

We could slightly change our model to make the interaction as a sum of binary interactions, replacing~$\bar{\omega}_k$ by~$\bar{J}_k$ (defined in equation~\eqref{omega_k_bar}) in the system of particles~\eqref{SDE_pos}-\eqref{SDE_orient}. 
In that case the model present an interesting phenomenon of phase transition and is the subject of current work in collaboration with P. Degond and J.-G. Liu~\cite{degond2011macroscopic} (the homogeneous version has been studied with him in~\cite{frouvelle2011dynamics}). 
In that case, writing the BBGKY hierarchy and using exchangeability of particles gives a system of evolution equations, the first one involving only the one-particle and the two-particles distributions. 
The classical assumption of propagation of chaos amounts to consider the two-particles density as the tensor of the one-particle density~$f$ by itself (in the limit of a large number of particles, two particles behave as if they were independent), and this gives exactly the evolution equation~\eqref{KFP}. Actually, it has been recently proved~\cite{bolley2011meanfield} that the mean-field limit of this model is the partial differential equation~\eqref{KFP}, where~$\bar{\omega}$ is replaced by~$\bar{J}$ in the definition of the force~$F$ (in the case where~$\nu$ and~$d$ are constant). The main point to derive this limit is to adapt the classical theory of propagation of chaos~\cite{mckean1967propagation,sznitman1991topics} in a framework of stochastic analysis in a Riemannian manifold (the unit sphere in the present case).

In the case of a non-linear diffusion, some results are given in~\cite{bolley2011stochastic} for other systems of self-propelled particles, under assumptions which would have to be adapted in our framework. 
We can expect to have conditions such as to be Lipschitz for the function~$\sqrt d$, and to be Lipschitz and bounded for the kernel~$\widetilde{K}$. 
Since we use the Stratonovich formulation in order to work on the unit sphere, we get  the term~$\nabla_\omega\cdot(\sqrt{d(\bar{\rho})}\nabla_\omega\sqrt{d(\bar{\rho})}f)$ (instead of~$\Delta_v(d(\bar{\rho})f)$ when the velocity~$v\in\mathbb{R}^d$ satisfies the SDE in the usual It\=o formulation, see the sections~$4.3.5$ and~$4.3.6$ of~\cite{gardiner1985handbook} for the correspondence). 

Finally, when the drift is not under the average form, it is sometimes possible to get a mean-field limit, under regularity assumptions on the coefficients~\cite{oelschlager1984martingale}, or with weaker assumptions, but assuming uniqueness of the solution of the mean-field model~\cite{nagasawa1987propagation}. These results could to be adapted in the framework of stochastic differential equations on the sphere, but dealing with the singularity of~$\bar{\omega}$ (when~$J$ is close to zero) seems to be slightly more complicated.

With these considerations in mind, it is however very reasonable that the limit of the particle system~\eqref{SDE_pos}-\eqref{SDE_orient}, when the number of particles is large, is given by the mean-field model~\eqref{KFP}. 
So we start with this model as a base for the derivation of the continuum model. 
A rigorous proof of the derivation of such a mean-field model from the individual dynamics is left to future work.
\begin{remark}
If we want to take into account some “ability to turn”, we just have to replace~$\nu(\bar{\rho})$ by~$\nu(\bar{\rho},\omega\cdot\bar{\omega})$.
\end{remark}
The next step consists in observing this system at large scale, in both space in time.

\subsection{Step 2: hydrodynamic scaling}
The hydrodynamic scaling consists in the same rescaling for the time and space variable. We introduce a small parameter~$\varepsilon$ and we set~$x'=\varepsilon x$, and~$t'=\varepsilon t$. We define~\mbox{$f^\varepsilon(x',\omega, t')=f(x,\omega,t)$}, and we rewrite the equation~\eqref{KFP} in this new coordinates.

The kinetic equation has the same form, with a factor~$\varepsilon$ in front of each of the terms with space or time derivative:
\begin{equation*}
\varepsilon(\partial_{t'} f^\varepsilon + \omega \cdot \nabla_{x'} f^\varepsilon) + \nabla_\omega \cdot (F^\varepsilon f^\varepsilon) = \nabla_\omega \cdot (\sqrt{d(\bar{\rho}^\varepsilon)}\nabla_\omega\sqrt{d(\bar{\rho}^\varepsilon)} f^\varepsilon) , 
\end{equation*}
with
\begin{equation*}
F^\varepsilon (x', \omega, t') = \nu(\bar{\rho}^\varepsilon) \, (\mathrm{Id} - \omega \otimes \omega) \bar{\omega}^\varepsilon (x',\omega,t') ,
\end{equation*}
where the local rescaled density and orientation are given by
\begin{align*}
\bar{\rho}^\varepsilon (x', \omega, t') &= \int_{ y \in \mathbb{R}^n , \, \upsilon \in \mathbb{S} } \widetilde K\left(|y-x'|,\tfrac{y-x'}{|y-x'|}\cdot\omega\right) \, f^\varepsilon(y, \upsilon,t') \,\tfrac{\mathrm d y}{\varepsilon^n} \, \mathrm d\upsilon\, ,\\
\bar{\omega}^\varepsilon (x', \omega, t') &= \frac{J^\varepsilon(x',\omega,t')}{|J^\varepsilon(x',\omega,t')|}, \\
J^\varepsilon(x',\omega,t') &= \int_{ y \in \mathbb{R}^n , \, \upsilon \in \mathbb{S} } K\left(\tfrac{|y-x'|}{\varepsilon},\tfrac{y-x'}{|y-x'|}\cdot\omega\right) \, 
\upsilon \, f^\varepsilon(y, \upsilon,t') \,\tfrac{\mathrm d y}{\varepsilon^n}\, \mathrm d\upsilon \, .
\end{align*}
The important point is to realize that the average density~$\bar{\rho}^\varepsilon$ and orientation~$\bar{\omega}^\varepsilon$ now depend on~$\varepsilon$, and can be easily expanded in terms of~$\varepsilon$, the non-locality only appearing at high order.
Omitting the primes for simplicity, we have the following expansions, the proofs of which are given in Lemma~\ref{lemma_expansion} of Appendix~\ref{annex_lemma_expansion}: 
\begin{align*} 
\bar{\omega}^\varepsilon (x, \omega, t) &= \Omega^\varepsilon (x, t) + \varepsilon \alpha\, (\omega\cdot\nabla_x)\, \Omega^\varepsilon(x,t) + O(\varepsilon^2)\, ,\\
\bar{\rho}^\varepsilon (x, \omega, t) &= \rho^\varepsilon (x, t) + \varepsilon \widetilde{\alpha}\, \omega\cdot\nabla_x\rho^\varepsilon(x,t) + O(\varepsilon^2)\, ,
\end{align*}
where~$\rho^\varepsilon=\rho_{f^\varepsilon}$ and~$\Omega^\varepsilon=\Omega_{f^\varepsilon}$ are the local density and mean orientation associated to the function~$f^\varepsilon$ (these quantities, depending only on the space and time variables, are related to the first moments with respect to the variable~$\omega$) given by
\begin{align} 
\rho_f &= \int_{ \omega \in \mathbb{S} } f(.,\omega) \, \mathrm d\omega \, ,\label{def_rho}\\
\Omega_f &= \frac{j_f}{| j_f |}, \text{ with } j_f = \int_{ \omega \in \mathbb{S} } \omega \, f(.,\omega) \, \mathrm d\omega \, ,\label{def_omega}
\end{align}
and the constants~$\alpha$ and~$\widetilde{\alpha}$ depend only on the observation kernels~$K$ and~$\widetilde K$. These constants are positive if the kernel is directed forward, and the more acute the “angle of vision”, the bigger the constant related to the kernel.

Now we can introduce these expansions in the mean-field model, and after some easy algebra, the rescaled model can be written in the form 
\begin{align} 
\varepsilon ( \partial_t f^\varepsilon + \omega \cdot \nabla_x f^\varepsilon + \alpha\, P(f^\varepsilon)+\widetilde{\alpha}\, \widetilde{P}(f^\varepsilon)) = Q(f^\varepsilon) + O(\varepsilon^2)\, , 
\label{KFP_rescaled} 
\end{align}
where~$Q$,~$P$ and~$\widetilde P$ are the operators given by the following equations (where~$\dot{\nu}$ and~$\dot d$ are the derivatives of~$\nu$ and~$d$ with respect to~$\rho$): 
\begin{align*} 
Q(f) =& \, - \nu(\rho_f)\nabla_\omega \cdot ((\mathrm{Id} - \omega \otimes \omega) \Omega_f f) + d(\rho_f)\Delta_\omega f , \\
P(f) =& \, \nu(\rho_f)\nabla_\omega \cdot ((\mathrm{Id} - \omega \otimes \omega)((\omega\cdot\nabla_x)\, \Omega_f) f),\\
\begin{split}
\widetilde{P}(f) =& \,\dot{\nu}(\rho_f) \nabla_\omega \cdot ((\omega\cdot\nabla_x\rho_f)\,(\mathrm{Id}-\omega\otimes\omega)\Omega_ff)\\
      &- \dot d(\rho_f) \nabla_\omega \cdot (\tfrac12(\mathrm{Id}-\omega\otimes\omega)\nabla_x\rho_f f+ (\omega\cdot\nabla_x\rho_f)\, \nabla_\omega f).
\end{split}
\end{align*}
Notice that the operator~$Q$ (giving the only term of order~$0$ in~$\varepsilon$) only acts on the variable~$\omega$, and the study of its properties will be important for the following.
\begin{remark}
If~$\nu$ also depends on~$\omega\cdot\bar{\omega}$, the expression of the operator~$Q$ is the same with~$\nu(\rho_f,\omega\cdot\Omega_f)$ instead of~$\nu(\rho_f)$. But then the expressions of~$P$ and~$\widetilde{P}$ complicate in a significant way, since there are also terms with the derivative of~$\nu$ with respect to this second variable.
\end{remark}
Now we are ready to study this system when~$\varepsilon\to0$.

\subsection{Step 3: limit as~$\varepsilon\to0$}

This is the main step, where we give the link between the continuum limit~\eqref{mass_eq}-\eqref{Omega_eq} and the rescaled kinetic equation~\eqref{KFP_rescaled} of the particle dynamics.

\begin{thm} 
The limit when~$\varepsilon \to 0$ of~$f^\varepsilon$ is given (formally) by~$f^0 = \rho M_{\kappa(\rho)\Omega}$ where~$\rho = \rho(x,t) > 0$ is the total mass of~$f^0$ and~$\Omega = \Omega(x,t) \in \mathbb{S}$ its mean orientation: 
\begin{gather*} 
\rho(x,t) = \int_{\omega \in {\mathbb{S}}} f^0(x,\omega, t) \, \mathrm d\omega , \\
\Omega = \frac{j}{|j|} \, , \quad j(x,t) = 
\int_{\omega \in {\mathbb{S}}} f^0(x,\omega, t) \, \omega \, \mathrm d\omega ,
\end{gather*}
where~$M_{\kappa\Omega}$ is a given function of~$\omega\cdot\Omega$ and ~$\kappa=\frac{\nu}d$ which will be specified later on (see~\eqref{M_def}). 
Furthermore,~$\rho(x,t)$ and~$\Omega(x,t)$ satisfy the following system of first order partial differential equations: 
\begin{gather*} 
\partial_t \rho + \nabla_x \cdot (c_1 \rho \Omega) = 0. \\
\rho \, \left( \partial_t \Omega + c_2 (\Omega \cdot \nabla_x) \Omega \right) + \lambda \, (\mathrm{Id} - \Omega \otimes \Omega) \nabla_x \rho = 0,
\end{gather*}
where the convection speeds~$c_1$,~$c_2$ and the parameter~$\lambda$ depend on~$\rho$. Their expressions will be given in this section (see~\eqref{def_c1} and~\eqref{def_c2}-\eqref{def_lambda}). 
\label{theo_limit}
\end{thm}

The method to obtain this result follows closely~\cite{degond2008continuum}, and is only summarized here. We will focus on the details which are specific to this study.

\subsubsection{Equilibria}

The first important point is to determine the null space~$\mathcal E$ of~$Q$, since it contains the limits of~\eqref{KFP_rescaled}. We find, as in~\cite{degond2008continuum}, that it is a~$n$-dimensional manifold consisting of functions analogous to Maxwellian distributions in the classical Boltzmann theory:
\begin{equation*} 
{\mathcal E} = \{ \rho M_{\kappa(\rho)\Omega}(\omega)\, | \, \rho > 0, \, \Omega \in \mathbb{S}\} \, ,
\end{equation*}
where
\begin{equation}
\kappa(\rho)=\frac{\nu(\rho)}{d(\rho)}>0\text{ and }M_{\kappa\Omega}(\omega) =\frac{e^{\kappa\, \omega \cdot \Omega}}{\int_\mathbb{S} e^{\kappa\, \upsilon \cdot \Omega}\, \mathrm d\upsilon}.\label{M_def} 
\end{equation}
The main difference with~\cite{degond2008continuum} is the dependence on~$\rho$ for this equilibrium in a nonlinear way, coming from the dependence of~$\nu$ and~$d$ on~$\rho$. This will result in additional terms in the computations, and so in additional terms in the expressions of the constants in the macroscopic model.

The normalization constant~$\int_\mathbb{S} e^{\kappa\, \omega \cdot \Omega}\, \mathrm d\omega$ depends only on~$\kappa$ (not on~$\Omega$) and so the total mass of~$M_{\kappa\Omega}(\omega)$ is~$1$ and its mean direction is~$\Omega$, that is to say~$\rho_{M_{\kappa\Omega}} = 1$ and~$\Omega_{M_{\kappa\Omega}}= \Omega$. Indeed we can easily compute the flux~$j_{M_{\kappa\Omega}}$ of this equilibrium, defined by~\eqref{def_omega}, and we get:
\begin{equation} 
 j_{M_{\kappa\Omega}} = \langle \cos \theta \rangle_{M_\kappa} \, \Omega , \label{flux}
\end{equation}
where for any function~$\gamma(\cos \theta)$, the notation~$\langle \gamma(\cos \theta) \rangle_{M_\kappa}$ stands for the mean of the function~$\omega\mapsto\gamma(\omega\cdot\Omega)$ against the density~$M_{\kappa\Omega}$, i.e. 
\begin{equation*} 
\langle \gamma(\cos \theta) \rangle_{M_\kappa} = \int_{\omega \in \mathbb{S}} M_{\kappa\Omega}(\omega) \gamma(\omega \cdot \Omega) \, \mathrm d\omega=\frac{\int_\mathbb{S} \gamma(\omega \cdot \Omega) e^{\kappa\, \omega \cdot \Omega}\, \mathrm d\omega}{\int_\mathbb{S} e^{\kappa\, \omega \cdot \Omega}\, \mathrm d\omega}.
\end{equation*}
Notice that~$\langle \gamma(\cos \theta) \rangle_{M_\kappa}$ depends only on~$\kappa$, not on~$\Omega$:
\begin{equation} 
\langle \gamma(\cos \theta) \rangle_{M_\kappa}= \frac{\int_0^\pi \gamma(\cos \theta) e^{\kappa\cos\theta} \, \sin^{n-2} \theta \, \mathrm d\theta}{\int_0^\pi e^{\kappa\cos\theta} \, \sin^{n-2} \theta \, \mathrm d\theta}.
\label{brackets}
\end{equation}
\begin{remark}
In the case where~$\nu$ depends on~$\rho$ and~$\omega\cdot\Omega$, we have to replace in all this point~$\kappa\, \omega\cdot\Omega$ by~$\widehat{\kappa}(\rho,\omega\cdot\Omega)$, where~$\widehat{\kappa}(\rho, \mu)={\int_0^\mu\frac{\nu(\rho,\tau)}{d(\rho)}\mathrm d\tau}$. 
\end{remark}

\subsubsection{Collisional invariants}

The second important point is the determination of generalized collisional invariants. Indeed, since there is no other conservation relation than the conservation of mass, the collision invariants reduce to the constants, and the integration of the equation against these invariants only gives one equation, which is not sufficient to describe the behavior of the equilibrium (which lives on a~$n$-dimensional manifold). 
The main idea in~\cite{degond2008continuum} was to overcome this problem with a generalization of the concept of collisional invariants.

A collision invariant is a function~$\psi$ such that for all function~$f$ of~$\omega$, the integration of~$Q(f)$ against~$\psi$ is zero. So we ask for a generalized invariant to satisfy this definition only for a restricted subset of functions~$f$. In the case where the dependence on~$\rho$ in the equilibria is linear, restricting to all functions with a given orientation~$\Omega$ is sufficient to obtain the remaining equation. Here we also have to restrict to functions with a given density too (actually, we impose a given~$\kappa(\rho)$). We will have then a set of generalized coefficients indexed by~$\Omega\in\mathbb{S}$ and~$\kappa>0$.

More precisely, to have a good definition, we have to work with linear operators (this point has been mentioned in~\cite{degond2010macroscopic}, since the result given in~\cite{degond2008continuum}, with the definition therein, was slightly incorrect).
We first define the linear operator~$L_{\kappa\Omega}$ by
\begin{equation*}
L_{\kappa\Omega}(f) = - \Delta_\omega f + \kappa\nabla_\omega\cdot((\mathrm{Id}-\omega\otimes\omega)\Omega f)= - \nabla_\omega \cdot \left[ M_{\kappa\Omega} \nabla_\omega \left( \frac{f}{M_{\kappa\Omega}} \right) \right],
\label{L_kappa_Omega}
\end{equation*}
 and then the generalized collisional invariants~${\mathcal C}_{\kappa\Omega}$ (associated to~$\kappa\in\mathbb{R}$ and~$\Omega\in\mathbb{S}$) as the following vector space:
\begin{equation*} 
{\mathcal C}_{\kappa\Omega}=\left\{\psi|\int_{\omega \in \mathbb{S}} L_{\kappa\Omega}(f) \, \psi \, \mathrm d\omega = 0 , \, \forall f \text{ such that } \, (\mathrm{Id}-\Omega\otimes\Omega)j_f=0 \right\}.
\end{equation*}
We remark that the operator~$Q(f)$ can be written as~$Q(f)= - d(\rho_f) L_{\kappa(\rho_f)\Omega_f}(f)$.
Hence, for any generalized collisional invariant~$\psi\in {\mathcal C}_{\kappa\Omega}$, we have
\begin{equation} 
\forall f \text{ such that } \, \Omega_f = \Omega \text{ and } \kappa(\rho_f)=\kappa, \int_{\omega \in \mathbb{S}} Q(f) \, \psi \, \mathrm d\omega = 0,
\label{Q_coll_invar}
\end{equation}
and this is the only property of generalized collisional invariants we will need in the following.

The computation of the set of generalized collisional invariants has been done in~\cite{degond2008continuum} in dimension~$3$. We give here the general result in any dimension.

\begin{prop} Structure of the generalized collisional invariants.

\label{structure_GCI}
Any generalized collisional invariant~$\psi$ associated to~$\kappa\in\mathbb{R}$ and~$\Omega\in\mathbb{S}$ has the following form:
\begin{equation*}
\psi(\omega)=C+h_\kappa(\omega\cdot\Omega)A\cdot\omega,
\end{equation*}
where~$C\in\mathbb{R}$ is a constant, the vector~$A\in\mathbb{R}^n$ is orthogonal to~$\Omega$, and~$h_\kappa$ is a given positive function on~$(-1,1)$, depending on the parameter~$\kappa$, which will be specified later on.
In particular, the generalized collisional invariants form a vector space of dimension~$n$.
\end{prop}

\begin{proof}
We first rewrite the set~$\{f | \, (\mathrm{Id}-\Omega\otimes\Omega)j_f=0 \}$ as the set of functions~$f$ such that for all~$A\in\mathbb{R}^n$ with~$A\cdot\Omega=0$, we have that~$\int_\mathbb{S}A\cdot\omega f \mathrm d\omega=0$. 
Finally this is the orthogonal of the set~$\{\omega\mapsto A\cdot\omega\text{ for } A\cdot\Omega=0 \}$, for the usual inner product on~$L^2(\mathbb{S})$. 
We can then rewrite the set of generalized collisional invariants:
\begin{align*} 
{\mathcal C}_{\kappa\Omega}&=\left\{\psi|\int_{\omega \in \mathbb{S}} f\,L_{\kappa\Omega}^* \, \psi \, \mathrm d\omega = 0 , \, \forall f \in\{\omega\mapsto A\cdot\omega\text{ for } A\cdot\Omega=0 \}^\perp \right\}\\
&=\{\psi|L_{\kappa\Omega}^* \, \psi \in(\{\omega\mapsto A\cdot\omega\text{ for } A\cdot\Omega=0 \}^\perp)^\perp\}\\
&=\{\psi|L_{\kappa\Omega}^* \, \psi (\omega)=A\cdot\omega \text{ with } A\cdot\Omega=0\},
\end{align*}
the operator~$L_{\kappa\Omega}^*$ being the adjoint of the operator~$L_{\kappa\Omega}$, which can be written
\begin{equation}
L_{\kappa\Omega}^*\,\psi = - \Delta_\omega\psi - \kappa\Omega\cdot\nabla_\omega \psi= - \frac1{M_{\kappa\Omega}} \nabla_\omega \cdot ( M_{\kappa\Omega} \nabla_\omega \psi).
\label{Lstar_kappa_Omega}
\end{equation}

It is then easy to show that the problem~$L_{\kappa\Omega}^* \, \psi (\omega)=A\cdot\omega$, for~$A\cdot\Omega=0$ has a unique solution in the space~$\dot H^1(\mathbb{S})$ (functions of~$H^1(\mathbb{S})$ with mean zero), using Lax-Milgram theorem and the Poincaré inequality. 
Hence, if we show that this solution has the form~$\psi(\omega)=h_\kappa(\Omega\cdot\omega)A\cdot\omega$, the solutions in the space~$H^1(\mathbb{S})$ are equal to this solution plus a constant~$C$. 

We search a solution of this form. 
We identify~$\Omega$ with the last element of an orthogonal basis of~$\mathbb{R}^n$, and~$\mathbb{S}_{n-2}$ with the elements on the unit sphere~$\mathbb{S}$ which are orthogonal to~$\Omega$. 
We can then write~$\omega=\cos\theta\,\Omega+\sin\theta\,v$, where~$v\in\mathbb{S}_{n-2}$ and~$\theta\in[0,\pi]$ (this decomposition is unique when~$\omega$ is different from~$\Omega$ and~$-\Omega$). In this framework we try to find a solution of the form~$\psi(\omega)=h_\kappa(\cos\theta)\sin\theta \, A\cdot v$.

For~$\psi(\omega)=g(\theta)Z(v)$, we have, in dimension~$n\geqslant3$:
\begin{equation*}
\nabla_\omega\psi(\omega)=g'(\theta) e_\theta Z(v) + \frac{g(\theta)}{\sin\theta}\nabla_vZ(v),
\end{equation*}
where the unit vector~$e_\theta$ is~$\nabla_\omega\theta=-\frac1{\sin\theta}(\mathrm{Id}-\omega\otimes\omega)\Omega$. 
A tangent vector field can always be written~$a\, e_\theta + \mathcal{A}$ where~$\mathcal{A}$ is a vector field tangent to the sphere~$\mathbb{S}_{n-2}$, and we have
\begin{equation*}
\nabla_\omega\cdot(a\, e_\theta + \mathcal{A})=\sin^{2-n}\theta \,\partial_\theta(\sin^{n-2}\theta\, a) + \frac1{\sin\theta}\nabla_v\cdot\mathcal{A}\,.
\end{equation*}
Finally we get, using the second part of~\eqref{Lstar_kappa_Omega},
\begin{equation*}
L_{\kappa\Omega}^*\,\psi=-\sin^{2-n}\theta e^{-\kappa\cos\theta}\tfrac{\mathrm d}{\mathrm d\theta}(\sin^{n-2}\theta e^{\kappa\cos\theta}g'(\theta))Z(v)-\tfrac{1}{\sin^2\theta}g(\theta)\Delta_vZ(v).
\end{equation*}
In our case we have~$Z(v)=A\cdot v$, so we get~$\Delta_vZ=-(n-2)Z$ (this is a spherical harmonic of degree~$1$ on~$\mathbb{S}_{n-2}$). So~$L_{\kappa\Omega}^*\,\psi$ is also of the form~$\widetilde L_\kappa^*g(\theta)Z(v)$, where 
\begin{equation}
\widetilde L_\kappa^*g(\theta)=-\sin^{2-n}\theta e^{-\kappa\cos\theta}\tfrac{\mathrm d}{\mathrm d\theta}(\sin^{n-2}\theta e^{\kappa\cos\theta}g'(\theta))+\tfrac{n-2}{\sin^2\theta}g(\theta).
\label{Ltild_star_kappa}
\end{equation}
Finally, solving~$L_{\kappa\Omega}^* \, (h(\omega\cdot\Omega)A\cdot\omega)=A\cdot\omega$ comes down to solving 
\begin{equation}
\widetilde L_\kappa^*g =\sin\theta,\quad\text{with}\quad g(\theta)=h(\cos\theta)\sin\theta.
\label{elliptic_problem}
\end{equation}
When~$A\neq0$, it easy to see that the function~$\omega\mapsto h(\omega\cdot\Omega)A\cdot\omega$ belongs to~$H^1(\mathbb{S})$ if and only if the function~$g:\theta\mapsto h(\cos\theta)\sin\theta$ belongs to the space~$V$ (a “weighted~$H^1_0$”) defined by
\begin{equation}
V = \{ g \, | \,(n-2)(\sin\theta)^{\frac n2-2} g \in L^2(0,\pi), \, (\sin\theta)^{\frac n2-1}g \in H^1_0(0,\pi) \}.
\label{def_V}
\end{equation}
Using again Lax-Milgram theorem in this space~$V$, we get that the problem~\eqref{elliptic_problem} has a unique solution, denoted~$g_\kappa$, which is positive (by the maximum principle). 
Writing~$h_\kappa(\mu)=(1-\mu^2)^{-\frac12}g_\kappa(\arccos(\mu))$ gives that~$\psi(\omega)=h_\kappa(\omega\cdot\Omega) \, A\cdot\omega$ is a solution to the partial differential equation~$L_{\kappa\Omega}^* \, \psi (\omega)=A\cdot\omega$. 
We could write the elliptic equation on~$(-1,1)$ satisfied by~$h_\kappa$ to have another definition, but this does not give a more elegant formulation.

In the case of dimension~$2$, we write~$\psi(\omega)=g(\theta)A\cdot v_0$, where~$v_0$ is one of the two unit vectors orthogonal to~$\Omega$ and~$g$ is an odd~$2\pi$-periodic function in~$H^1_{loc}(\mathbb{R})$, which can be identified with a function~$g\in H^1_0(0,\pi)=V$. We still have that the elliptic problem~$L_{\kappa\Omega}^* \, \psi (\omega)=A\cdot\omega$ is equivalent to~\eqref{elliptic_problem} with~$g\in V$, with the same definitions~\eqref{Ltild_star_kappa}-\eqref{def_V} of~$\widetilde L_\kappa$ and~$V$.
But since this elliptic equation reduces to~$( e^{\kappa\cos\theta}g'(\theta))'=-\sin\theta\, e^{\kappa \cos\theta}$, we now have the following explicit expression of~$g_\kappa$:
\begin{equation}
g_\kappa(\theta)=\frac{\theta}{\kappa}-\frac{\pi}{\kappa}\frac{\int_0^\theta e^{-\kappa \cos\varphi}\mathrm d\varphi}{\int_0^\pi e^{-\kappa \cos\varphi}\mathrm d\varphi}\, .
\label{def_g2D}
\end{equation}

\end{proof}

\begin{remark}
If we take into account the “ability to turn”, we just replace~$\kappa\cos\theta$ in equation~\eqref{Ltild_star_kappa} by~$\widehat{\kappa}(\cos\theta)$. In dimension~$2$, we still have an explicit expression:
\begin{equation}
\label{g_turn} g_{\widehat{\kappa}}(\theta)=g_{\widehat{\kappa}}^0(\theta)-\tfrac{g_{\widehat{\kappa}}^0(\pi)}{g_{\widehat{\kappa}}^\infty(\pi)}g_{\widehat{\kappa}}^\infty(\theta),
\end{equation}
where
\begin{gather}
\label{g_0_turn} g_{\widehat{\kappa}}^0(\theta)=-\int_0^\theta\int_\varphi^\pi e^{\widehat{\kappa}(\cos\phi)-\widehat{\kappa}(\cos\varphi)} \sin\phi\, \mathrm d\phi\, \mathrm d\varphi,\\
\label{g_inf_turn} g_{\widehat{\kappa}}^\infty(\theta)=\int_0^\theta e^{-\widehat{\kappa}(\cos\varphi)}\, \mathrm d\varphi\, .
\end{gather}
\end{remark}

\subsubsection{Computation of the limit as~$\varepsilon\to0$}
The third and final important point is taking the limit~$\varepsilon\to0$ in the equation~\eqref{KFP_rescaled}, after integrating against the collision invariants. Since we do not have results of existence, uniqueness and regularity of the solution, all the limits in this section have to be understood as formal limits. A rigorous proof of convergence is left to future work.

When~$\varepsilon\to0$, if we fix~$x$ and~$t>0$, we have that~$Q(f^\varepsilon)$, as a function of~$\omega$, tends formally to zero, so~$f^\varepsilon$ tends to an equilibrium of the operator~$Q$, of the form~$\rho M_{\kappa(\rho)\Omega}$, where~$\rho>0$ and~$\Omega\in\mathbb{S}$ are given functions of~$x$ and~$t$. This is the first part of Theorem~\ref{theo_limit}. So we have~$\rho^\varepsilon\to\rho$, and~$\Omega^\varepsilon\to\Omega$. When there is no possible confusion, we will write~$\kappa$ for~$\kappa(\rho)$.

For the mass equation, we use the constant invariant: we have, since the operators~$Q$,~$P$ and~$\widetilde{P}$ are given as the divergence (with respect to~$\omega$) of a function,
\begin{equation*}
\int_{\omega \in \mathbb{S}}Q(f^\varepsilon)\, \mathrm d\omega=\int_{\omega \in \mathbb{S}}P(f^\varepsilon)\, \mathrm d\omega=\int_{\omega \in \mathbb{S}}\widetilde{P}(f^\varepsilon)\, \mathrm d\omega=0.
\end{equation*}
Hence, integrating the equation~\eqref{KFP_rescaled} with respect to~$\omega$, we get: 
\begin{equation*} 
\partial_t \rho^\varepsilon + \nabla_x \cdot j^\varepsilon = O(\varepsilon).
\end{equation*}
Actually, we can even replace the~$O(\varepsilon)$ by zero in this equation since in the original model~\eqref{KFP} we have conservation of mass. We get in the~$\varepsilon\to0$ limit: 
\begin{equation*} 
\partial_t \rho + \nabla_x \cdot (c_1(\kappa) \rho\, \Omega) = 0,
\end{equation*}
where (see~\eqref{flux}):
\begin{equation}
c_1(\kappa) = |j_{M_{\kappa\Omega}}|=\langle \cos \theta \rangle_{M_\kappa}=\tfrac{\int_0^\pi\cos \theta e^{\kappa\cos\theta} \, \sin^{n-2} \theta \, \mathrm d\theta}{\int_0^\pi e^{\kappa\cos\theta} \, \sin^{n-2} \theta \, \mathrm d\theta}.
\label{def_c1}
\end{equation}
This gives the second part of Theorem~\ref{theo_limit}, with the equation on~$\rho$ and the definition of~$c_1$.

To get the equation on~$\Omega$, we use the non-constant part of the collisional invariants.
By Proposition~\ref{structure_GCI} and using the result at equation~\eqref{Q_coll_invar}, we get that for all~$A$ such that~$A\cdot\Omega^\varepsilon=0$, we have
\begin{equation*}
\int_{\omega \in \mathbb{S}}Q(f^\varepsilon)h_{\kappa(\rho^\varepsilon)}(\omega\cdot\Omega^\varepsilon) \,A\cdot\omega\, \mathrm d\omega=0.
\end{equation*}
So we have that the vector~$X^\varepsilon=\int_{\omega \in \mathbb{S}}Q(f^\varepsilon)h_{\kappa(\rho^\varepsilon)}(\omega\cdot\Omega^\varepsilon)\, \omega\, \mathrm d\omega$ is orthogonal to~$A$ for all~$A$ orthogonal to~$\Omega$, that is to say that~$X^\varepsilon$ is in the direction of~$\Omega^\varepsilon$, which is equivalent to~$(\mathrm{Id} - \Omega^\varepsilon \otimes \Omega^\varepsilon)\, X^\varepsilon=0$.
Using~\eqref{KFP_rescaled}, we get that
\begin{equation*} 
X^\varepsilon = \varepsilon\int_{\omega \in \mathbb{S}} ( \partial_t f^\varepsilon + \omega \cdot \nabla_x f^\varepsilon+\alpha P(f^\varepsilon) +\widetilde{\alpha} \widetilde{P}(f^\varepsilon))\, h_{\kappa(\rho^\varepsilon)}(\omega \cdot \Omega^\varepsilon) \, \omega \, \mathrm d\omega\, +O(\varepsilon^2).
\end{equation*}
Dividing by~$\varepsilon$ and taking the limit~$\varepsilon\to0$, we get~$(\mathrm{Id} - \Omega \otimes \Omega)\, X=0$, where 
\begin{equation} 
X = \int_{\omega \in \mathbb{S}} ( \partial_t ( \rho M_{\kappa\Omega}) + \omega \cdot \nabla_x (\rho M_{\kappa\Omega})+\alpha P(\rho M_{\kappa\Omega}) +\widetilde{\alpha} \widetilde{P}(\rho M_{\kappa\Omega}))\, h_\kappa(\omega \cdot \Omega) \, \omega \, \mathrm d\omega\, .
\label{def_X}
\end{equation}

The main point is then to compute~$(\mathrm{Id} - \Omega \otimes \Omega)\, X$, in terms of~$\rho$,~$\Omega$ and their derivatives, using mainly the chain rule. The computation is similar to~\cite{degond2008continuum} for some terms, but additional work is required for the terms coming from the nonlinearity of~$M_{\kappa\Omega}$ in~$\rho$ and the operators~$P$ and~$\widetilde{P}$. 
We give the result of the computations under the form of a proposition: 
\begin{prop}$(\mathrm{Id} - \Omega \otimes \Omega)\, X=0$, where~$X$ is given in~\eqref{def_X}, is equivalent to 
\label{prop_dtOmega}
\begin{equation*}
\rho \, \left( \partial_t \Omega + c_2 (\Omega \cdot \nabla_x) \Omega \right) + \lambda \, (\mathrm{Id} - \Omega \otimes \Omega) \nabla_x \rho = 0,
\end{equation*}
where
\begin{align}
c_2&=\widetilde c_1 - \alpha \,d\,(n\, \widetilde c_1+ \kappa \,\langle \cos^2\theta\rangle_{\widetilde M_\kappa})\, ,\text{ with} \quad \widetilde c_1=\langle \cos\theta \rangle_{\widetilde M_\kappa}\, ,\label{def_c2}\\
\lambda&=\tfrac1{\kappa}+\rho\,\tfrac{\dot{\kappa}}{\kappa}\,[ \,\widetilde c_1 -c_1+\widetilde{\alpha}\,d\,(\kappa \,\langle \sin^2\theta\rangle_{\widetilde M_\kappa}-n\, \widetilde c_1) \,] + \tfrac12\,\widetilde{\alpha}\,\rho\,\dot{d}\,(n-1+\kappa \,\widetilde c_1)\, ,\label{def_lambda}
\end{align}
with the notation
\begin{equation*} 
\langle \gamma(\cos \theta) \rangle_{\widetilde M_\kappa}=
 \frac{\int_0^\pi \gamma(\cos \theta) h_\kappa(\cos\theta) e^{\kappa\cos\theta} \, 
\sin^{n} \theta \, \mathrm d\theta}{\int_0^\pi h_\kappa(\cos\theta)e^{\kappa\cos\theta} \, \sin^{n} \theta \, \mathrm d\theta}. 
\end{equation*}
\end{prop}
This proposition is exactly the last part of Theorem~\ref{theo_limit}, with a precise definition for coefficients~$c_2$ and~$\lambda$, and this ends the derivation of the continuum model~\eqref{mass_eq}-\eqref{Omega_eq}.

The computations to get this result are given in Appendix~\ref{annex_dtOmega}, the idea is to compute~$(\mathrm{Id} - \Omega \otimes \Omega)\, X$ using the chain rule and the change of variable~$\omega\rightsquigarrow(\theta,v)$ where~$\omega=\cos\theta\,\Omega+\sin\theta\,v$, with~$v$ orthogonal to~$\Omega$, which simplifies a lot of terms.

\begin{remark}
The computations have also been done in the case where~$\nu$ depends on~$\omega\cdot\Omega$ (and not on~$\rho$) and where~$d$ is a constant. We get the same results, except that the constants are given (with analogous definitions) by:
\begin{align}
c_1 =& \left\langle \cos \theta\right\rangle_{M_{\widehat{\kappa}}} \label{c_1_3d}\\
\begin{split}
c_2 =& \langle \cos \theta\rangle_{\widetilde M_{\widehat{\kappa}}}-\alpha \langle\nu\cos^2\theta-\nu'\cos\theta\sin^2\theta\rangle_{\widetilde M_{\widehat{\kappa}}} \\ 
& - \, \alpha d \left\langle n \cos\theta + \tfrac{\nu'}{\nu} ((n+2)\cos^2\theta-1) - \tfrac{\nu''}{\nu}\cos\theta\sin^2 \theta
 \right\rangle_{\widetilde M_{\widehat{\kappa}}},
\label{c2_3d}
\end{split}\\
\lambda =& d\left\langle \tfrac{1}{\nu} \right\rangle_{\widetilde M_{\widehat{\kappa}}}\label{lambda_3d},
\end{align}
where here we use the notation
\begin{equation*} 
\langle \gamma(\cos \theta) \rangle_{\widetilde M_{\widehat{\kappa}}}=
 \frac{\int_0^\pi \gamma(\cos \theta) \nu(\cos\theta) h_{\widehat{\kappa}}(\cos\theta) e^{\widehat{\kappa}(\cos\theta)} \, 
\sin^{n} \theta \, \mathrm d\theta}{\int_0^\pi \nu(\cos\theta) h_{\widehat{\kappa}}(\cos\theta) e^{\widehat{\kappa}(\cos\theta)} \, \sin^{n} \theta \, \mathrm d\theta}. 
\end{equation*}

Since~$\nu$ is supposed to be positive, the constant~$\lambda$ is positive, and we will see in the next section that its possible change of sign with the dependence on~$\rho$ is important. This is why we focus on the dependence on~$\rho$ and not in~$\omega\cdot\Omega$ in this article.

\end{remark}

\section{Properties of the macroscopic model}
\label{properties}
\subsection{Hyperbolicity}
\label{hyperbolicity}

We recall here the macroscopic model~\eqref{mass_eq}-\eqref{Omega_eq}:
\begin{gather*}
\partial_t \rho + \nabla_x \cdot (c_1(\rho) \rho \Omega) = 0,\\
\rho \, \left( \partial_t \Omega + c_2(\rho) (\Omega \cdot \nabla_x) \Omega \right) + \lambda(\rho) \, (\mathrm{Id} - \Omega \otimes \Omega) \nabla_x \rho = 0,
\end{gather*}
where the functions~$c_1$,~$c_2$, and~$\lambda$ are given by~\eqref{def_c1} and~\eqref{def_c2}-\eqref{def_lambda}.
A first remark is that it is not possible to do another scaling to get rid of~$c_1$, like in~\cite{degond2008continuum}, because~$c_1$ depends on~$\rho$.

The main result about this model is that if~$d$ or~$\nu$ depends on~$\rho$, the coefficient~$\lambda$ can become negative in some regions of the state space, and in that case the system loses hyperbolicity. Let us first discuss here the interest and the problems due to the non-hyperbolicity.

The first thing to remark is that the model is not always well-posed. Indeed, in general, we cannot ensure that a solution will stay in the region of hyperbolicity for all time, even with smooth initial conditions in the hyperbolic region (actually, even with hyperbolicity everywhere, dealing with the discontinuities is a challenging issue, see~\cite{motsch2010numerical}).

The property of hyperbolicity is linked with the fact that perturbations propagate with finite speed. Here the presence of a region of non-hyperbolicity means that we could have propagation with infinite speed across this region.
This leads to a second remark: it may be possible to construct non-classical shocks, using the crossing of a zone of non-hyperbolicity, see~\cite{lefloch2002hyperbolic}, and~\cite{keyfitz1993multiphase}.
The interest is that we may construct some travelling waves, as observed in~\cite{chate2008collective}. Actually we did not manage to construct such solutions yet, this is part of our future work.

We should also construct models with formation of coherent structures from such non-hyperbolic models, if we could use stabilization with diffusion. But here the expansion at higher order in~$\varepsilon$ in the rescaled mean-field model~\eqref{KFP_rescaled}, in order to obtain diffusion terms in the macroscopic model becomes too much complicated to perform some study (see~\cite{degond2010diffusion} for the case of the original model of~\cite{degond2008continuum}).

We now turn to the description of the regions of non-hyperbolicity.
We consider a system satisfying~\eqref{mass_eq}-\eqref{Omega_eq}, but evolving only along one space direction~$e_z\in\mathbb{S}$ (the density~$\rho$ and the orientation~$\Omega$ depending only on~$z=e_z\cdot x$ and~$t$). We write then~$\Omega=\cos\theta\,e_z+\sin\theta\,v$,where~$v\in\mathbb{S}_{n-2}$ (identified to the set of unit vectors orthogonal to~$e_z$). In this framework, the system is equivalent to
\begin{gather} 
\partial_t \rho + \, \partial_z (\rho c_1(\rho) \cos \theta) = 0.
\label{rho_eq_z} \\
\rho[\partial_t (\cos\theta) + c_2(\rho) \cos \theta \, \partial_z (\cos\theta)] + \lambda \, (1-\cos^2 \theta) \, \partial_z \rho = 0.
\label{theta_eq_z} \\
\partial_t v + c_2(\rho) \cos \theta \, \partial_z v = 0, \text{ with }|v|=1 \text{ and } e_z\cdot v=0.
\label{v_eq_z}
\end{gather}
In the special case of dimension~$2$, the system reduces to~\eqref{rho_eq_z}-\eqref{theta_eq_z}, with~$\theta\in(-\pi,\pi)$ and~$\Omega=\cos\theta\,e_z+\sin\theta\,v_0$, where~$v_0$ is one of the two unit vectors orthogonal to~$e_z$.

The general definition of a quasilinear hyperbolic system~\cite{serre1996systemes} gives that the system~\eqref{mass_eq}-\eqref{Omega_eq} is hyperbolic if and only if this system~\eqref{rho_eq_z}-\eqref{v_eq_z} is hyperbolic for all unit vector~$e_z\in\mathbb{S}$. We give the result in the following statement:

\begin{thm} Hyperbolicity.
\label{theo_hyper}
\begin{itemize}
\item The system~\eqref{mass_eq}-\eqref{Omega_eq} is hyperbolic if and only if~$\lambda(\rho)> 0$.
\item The system~\eqref{rho_eq_z}-\eqref{v_eq_z} is hyperbolic if and only if
\begin{equation}
\lambda(\rho)> 0\text{ or }
\begin{cases}
|\tan\theta|<\frac{|c_2-c_3|}{2\sqrt{-\lambda c_1}},&\text{ if } \lambda<0 ,
\label{cond_hyp}\\
\theta\neq\frac{\pi}2\text{ and }c_2\neq c_3,&\text{ if }\lambda=0.
\end{cases}
\end{equation} 
where~$c_3(\rho)=\frac{\mathrm d}{\mathrm d\rho}(\rho c_1(\rho))=c_1(\rho)+\rho\dot{\kappa}\left(\langle \cos^2 \theta \rangle_{M_\kappa}-\langle \cos \theta \rangle_{M_\kappa}^2\right)$.
\end{itemize}
\end{thm}

\begin{proof}
The system~\eqref{rho_eq_z}-\eqref{v_eq_z} can be written as the following first order quasilinear system of partial differential equations
\begin{equation*} 
\left( \begin{array}{c} \partial_t \rho \\ \partial_t\cos\theta \\ \partial_t v \end{array} \right) 
+ A(\rho, \cos\theta, v) \left( \begin{array}{c} \partial_z \rho \\ \partial_z \cos\theta \\ \partial_z v \end{array} \right) = 0,
\end{equation*}
with 
\begin{equation*} 
A(\rho, \cos\theta, v) = \left( \begin{array}{ccc} 
c_3(\rho)\cos \theta & c_1(\rho)\rho & 0\quad \cdots \quad 0 \\ 
\frac{\lambda}{\rho} \sin^2 \theta & c_2(\rho) \cos \theta & 0\quad \cdots \quad 0 \\
0 & 0 & \\
\vdots & \vdots & c_2(\rho) \cos \theta\, \mathrm{Id}_{n-2}\\
0 & 0 & \end{array} \right), 
\end{equation*}
and this system is hyperbolic in case~$\lambda>0$. The eigenvalues are~$\gamma_\pm$ and~$\gamma_0$ (of multiplicity~$n-2$), given by
\begin{equation*} 
\gamma_0 = c_2\cos \theta,\quad \gamma_\pm=\frac{1}{2} \left[(c_2+c_3)\cos\theta\pm\left( (c_2-c_3)^2 \cos^2 \theta+ 4\lambda c_1\sin^2 \theta\right)^{1/2}\right], 
\end{equation*}
Now if~$\lambda<0$, asking~$\gamma_\pm$ to be real and distinct is exactly equivalent to the equation~\eqref{cond_hyp}. In this case the matrix~$A$ is diagonalizable. If~$\gamma_+=\gamma_-$, then~$A$ is diagonalizable only if its top left corner~$2\times2$ submatrix is scalar (only one eigenvalue), which is not the case since~$c_1(\rho)\rho>0$. For the same reason, if~$\lambda=0$, we immediately get that~$A$ is diagonalizable if and only if the first two diagonal coefficients~$c_2(\rho) \cos \theta$ and~$c_3(\rho) \cos \theta$ are different, which ends the proof of the second statement.

Now we turn to the general case. If~$\lambda>0$, the system~\eqref{rho_eq_z}-\eqref{v_eq_z} is hyperbolic for all unit vector~$e_z\in\mathbb{S}$, which gives that the system~\eqref{mass_eq}-\eqref{Omega_eq} is hyperbolic.
Suppose now that the system~\eqref{mass_eq}-\eqref{Omega_eq} is hyperbolic in an open region of the state space, with~$\lambda\leqslant0$ at some point~$(\rho,\Omega)$. 
Since~$n\geqslant2$ we can find~$e_z$ such that~$e_z\cdot\Omega=0$. 
Then we have~$\cos\theta=0$, which gives, by the condition~\eqref{cond_hyp} that~$(\rho,\Omega)$ is in the region of non-hyperbolicity of the problem~\eqref{rho_eq_z}-\eqref{v_eq_z}, and this is a contradiction. 
\end{proof}

Actually, the positive functions~$d$ and~$\nu$ being arbitrary, it is possible to have a lot of qualitatively different shapes for the region of non-hyperbolicity of the reduced system~\eqref{rho_eq_z}-\eqref{v_eq_z}. We give here some examples in the case of dimension~$2$, where the coefficients are easy to compute numerically (using the explicit formulation~\eqref{def_g2D} for~$g_\kappa$).
\begin{figure}[h!]
\begin{center}
\input{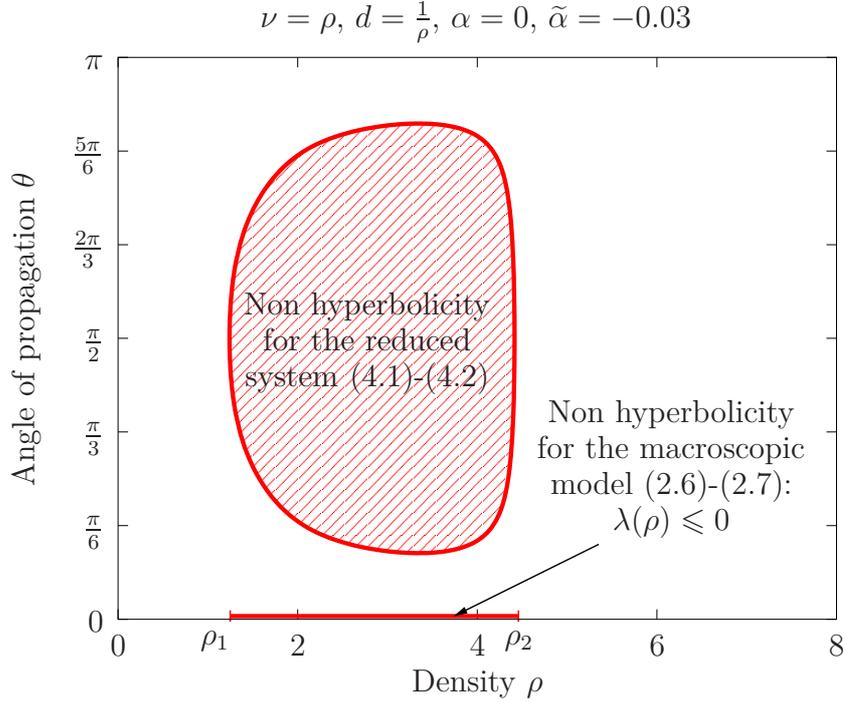}
\caption{The region of hyperbolicity can be of the form~$(0,\rho_1)\cup(\rho_2,+\infty)$}
\end{center}
\end{figure}

\begin{figure}[h!]
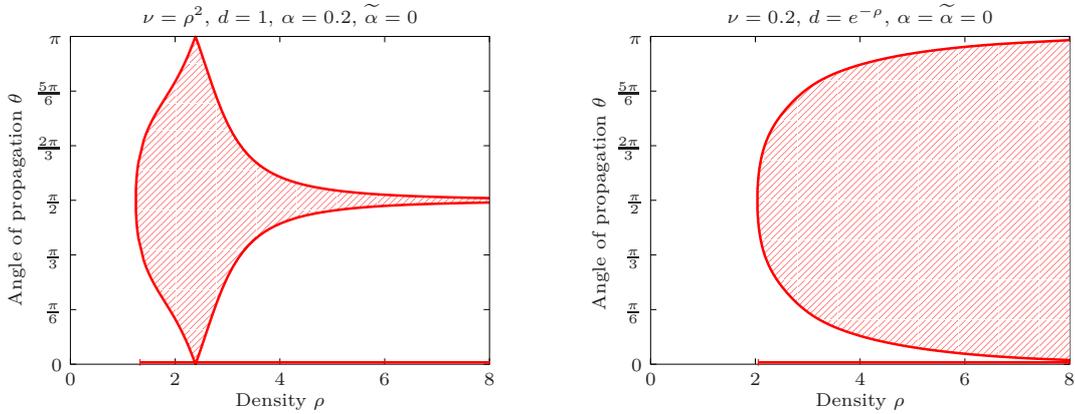

\begin{minipage}[h]{0.98\linewidth}
\begin{minipage}{0.48\linewidth}
\begin{center}
\input{zone2.pstex_t}
\end{center}
\end{minipage}\hfill
\begin{minipage}{0.48\linewidth}
\begin{center}
\input{zone3.pstex_t}
\end{center}
\end{minipage}
\end{minipage}
\caption{Other shapes for the zone of non-hyperbolicity}
\end{figure}
\newpage

\subsection{Influence of the anisotropy}

On the final macroscopic model, the influence of the anisotropy in the observation kernels is only visible through the values of the speed~$c_2$, and the coefficient~$\lambda$. 

We remark that the parameter~$\alpha$, which is related to the kernel~$K$ used to define the local orientation~$\bar{\omega}$, only appears in the expression~\eqref{def_c2} of the velocity~$c_2$, making it smaller when~$\alpha$ is a large positive constant. 
The difference between~$c_1$ and~$c_2$, which is one of the differences between the macroscopic model~\eqref{mass_eq}-\eqref{Omega_eq} and the classical Euler system, is then enhanced when~$\alpha$ is a large positive constant. 
This can be interpreted as follows: if the observation kernel is strongly directed forward, then the information on the orientation moves rapidly backward. 
This could be compared to results on modelling of traffic flows, where the speed of a congested phase depends on the distance of anticipation of the drivers (see~\cite{aw2000resurrection,berthelin2008model,daganzo1995requiem}).

The parameter~$\widetilde{\alpha}$, related to the kernel~$\widetilde K$ which is used to define the local density~$\bar{\rho}$, only appears in~$\lambda$, and obviously has an influence only if the relaxation frequency~$\nu$ or the noise intensity~$d$ depends on this density (in the expression~\eqref{def_lambda}, we must have~$\dot{\kappa}\neq0$ or~$\dot d\neq0$).
So the anisotropy of the kernel~$\widetilde K$ can have an impact on the region of non-hyperbolicity for the system~\eqref{mass_eq}-\eqref{Omega_eq}. 
The anisotropy of the kernel~$K$ does not play a role in this global hyperbolicity, but, through the condition~\eqref{cond_hyp}, it can change the shape of the region of non-hyperbolicity for the one-dimensional reduction~\eqref{rho_eq_z}-\eqref{v_eq_z}. 

Since the expression~\eqref{def_lambda} involves a lot of terms which can take different signs, it is not easy to directly quantify the influence of the parameter~$\widetilde{\alpha}$ on the coefficient~$\lambda$, as it was the case for~$\alpha$ and~$c_2$.
In the next section, we perform an asymptotic study of the coefficients of the macroscopic model~\eqref{mass_eq}-\eqref{Omega_eq}, as the concentration parameter~$\kappa$ tends to infinity (in the case of strong alignment, or low noise) or to zero (when the noise is high, or the frequency of alignment).

\section{Asymptotic study of the coefficients}

\label{asymptotics}
We want to obtain an asymptotic expansion of~$c_1$,~$c_2$ and~$\lambda$ given by the expressions~\eqref{def_c1},~\eqref{def_c2},~\eqref{def_lambda} as the parameter~$\kappa$ tends to infinity or to zero.

Since we do not know explicitly the dependence on~$\rho$ for the coefficients~$\nu$ and~$d$, the only quantities we can study in the expressions of this coefficients are the averages~$c_1=\langle \cos \theta \rangle_{M_{\kappa}}$,~$\widetilde{c_1}=\langle \cos\theta \rangle_{\widetilde M_{\kappa}}$, and~$\langle\cos^2\theta\rangle_{\widetilde M_{\kappa}}$ (since the last average is~$\langle\sin^2\theta\rangle_{\widetilde M_{\kappa}}=1-\langle\cos^2\theta\rangle_{\widetilde M_{\kappa}}$).
The purpose of this section is to give a method to get the Taylor expansion up to any order in~$\kappa$ or~$\frac1{\kappa}$ of the following averages:
\begin{align*}
\langle f(\theta) \rangle_{M_{\kappa}}&= \frac{\int_0^\pi  f(\theta) e^{\kappa\cos\theta} \, \sin^{n-2} \theta \, \mathrm d\theta}{\int_0^\pi e^{\kappa\cos\theta} \, \sin^{n-2} \theta \, \mathrm d\theta},\\
\langle f(\theta)  \rangle_{\widetilde M_{\kappa}}&= \frac{\int_0^\pi  f(\theta) h_\kappa(\cos\theta) e^{\kappa\cos\theta} \, \sin^{n} \theta \, \mathrm d\theta}{\int_0^\pi h_\kappa(\cos\theta) e^{\kappa\cos\theta} \, \sin^{n} \theta \, \mathrm d\theta},
\end{align*}
where~$h_\kappa$ is the function providing the generalized collisional invariants (see Proposition~\ref{structure_GCI}).
We first give the method to obtain the expansion of the first type of average, and we apply it to get an expansion of~$c_1$ in~$\kappa$ and~$\frac1{\kappa}$.
\subsection{Asymptotics of~$\langle f(\theta) \rangle_{M_{\kappa}}$}
The first expansion, when~$\kappa\to0$, is just a basic Taylor expansion. For a function~$f$ such that~$f\sin^{n-2}\theta\in L^1(0,\pi)$, we define
\begin{equation*}
b_p=\frac1{p!}\int_0^\pi  f(\theta) \cos^p\theta \, \sin^{n-2} \theta \, \mathrm d\theta \quad \text{and}\quad  a_p=\frac1{p!}\int_0^\pi \cos^p\theta \, \sin^{n-2} \theta \, \mathrm d\theta.
\end{equation*}
Then we get
\begin{equation}
\langle f(\theta) \rangle_{M_{\kappa}}= \frac{\sum_{p=0}^Nb_p\kappa^p+O(\kappa^{N+1})}{\sum_{p=0}^Na_p\kappa^p+O(\kappa^{N+1})}.
\label{exp_f_zero}
\end{equation}

If we take~$f(\theta)=\cos\theta$, we have~$b_p=(p+1)a_{p+1}$ and integrating by parts, we get the following induction relation:~$(p+2)(p+n)a_{p+2}=a_p$. Since~$a_1=0$, the odd terms vanish and we get
\begin{equation}
c_1=\langle \cos \theta \rangle_{M_{\kappa}}= \frac{\frac1n \kappa+\frac1{2n(n+2)} \kappa^3+O(\kappa^5)}{1+\frac1{2n} \kappa^2+O(\kappa^4)}=\tfrac1n \kappa-\tfrac1{n^2(n+2)} \kappa^3+O(\kappa^5).\label{exp_c1_zero}
\end{equation}

We now turn to the expansion of~$c_1$ when~$\kappa\to\infty$.
We will use the following lemma, the proof of which is elementary, see~\cite{bender1999advanced} for examples and variants:
\begin{lemma} (Watson’s Lemma)
\label{watson}

Let~$p$ be a function in~$L^1(0,T)$, with~$T>0$, and let~$I_\kappa(p)=\int_0^Tp(t)e^{-\kappa t}\mathrm d t$. Suppose that, in the neighborhood of~$0$, we have~$p(t)=t^\beta\left(\sum_{i=0}^{N-1}a_it^i+O(t^N)\right)$, with~$\beta>-1$.

Then~$I_\kappa(p)=\kappa^{-\beta-1}\left(\sum_{i=0}^{N-1}a_i\Gamma(\beta+i+1)\kappa^{-i}+O(\kappa^{-N})\right)$ as~$\kappa\to\infty$.
\end{lemma}

We use this lemma, after the change of variable~$t=1-\cos\theta$, in the integrals of the form~$[f(\theta)]_\kappa=\int_0^\pi  f(\theta) e^{\kappa\cos\theta} \, \sin^{n-2} \theta \, \mathrm d\theta$. We get
\begin{equation*}
[f(\theta)]_\kappa=e^\kappa\int_0^2 f(arccos(1-t))e^{-\kappa t}(2t-t^2)^{\frac{n-3}2} \, \mathrm d t.
\end{equation*}
So if we can expand the function~$t\mapsto(2t-t^2)^{\frac{n-3}2}f(arccos(1-t))$ in the neighborhood of~$0$, we can apply directly Watson’s Lemma to get an expansion of~$[f(\theta)]_\kappa$, and then to~$[1]_\kappa$, which gives finally the expansion of~$\langle f(\theta)\rangle_{M_\kappa}$. 

We take here the example of the function~$f(\theta)=1-\cos \theta$, so~$f(arccos(1-t))=t$. We want an expansion with two terms (since we have~$c_1=1-\langle f(\theta)\rangle_{M_\kappa}$ we will actually get three terms for~$c_1$). 
We have 
\begin{equation*}
(2t-t^2)^{\frac{n-3}2}=2^{\frac{n-3}2}t^{\frac{n-3}2}(1-\tfrac12 t)^{\frac{n-3}2}=2^{\frac{n-3}2}t^{\frac{n-3}2}(1-\tfrac{n-3}4t +O(t^2))\end{equation*}
Applying directly Watson’s Lemma to this function and to the same function multiplied by~$t$, we get
\begin{align}
[1]_\kappa&=\frac{2^{\frac{n-3}2}e^\kappa}{\kappa^{\frac{n-1}2}}\left(\Gamma(\tfrac{n-1}2)-\tfrac{n-3}4\Gamma(\tfrac{n+1}2)\frac1{\kappa}+O(\kappa^{-2})\right)\label{int_exp}\\
[f(\theta)]_\kappa&=\frac{2^{\frac{n-3}2}e^\kappa}{\kappa^{\frac{n+1}2}}\left(\Gamma(\tfrac{n+1}2)-\tfrac{n-3}4\Gamma(\tfrac{n+3}2)\frac1{\kappa}+O(\kappa^{-2})\right)\nonumber.
\end{align}

Since~$\Gamma(p+1)=p\Gamma(p)$, we finally get
\begin{align*}
\langle f(\theta)\rangle_{M_{\kappa}}=\frac{[f(\theta)]_\kappa}{[1]_\kappa}&=\frac{\Gamma(\tfrac{n+1}2)}{\kappa\Gamma(\tfrac{n-1}2)}\frac{1-\tfrac{n-3}4\tfrac{n+1}2\frac1{\kappa}}{1-\tfrac{n-3}4\tfrac{n-1}2\frac1{\kappa}}+O(\kappa^{-3})\\
&= \frac{n-1}{2\kappa} -\frac{(n-1)(n-3)}{8\kappa^2}+ O(\kappa^{-3}).
\end{align*}
In particular we get the expansion of~$c_1$ as~$\kappa\to\infty$:
\begin{equation}
c_1=1- \frac{n-1}{2\kappa} +\frac{(n-1)(n-3)}{8\kappa^2} + O(\kappa^{-3}).
\label{exp_c1_inf}
\end{equation}
Using this method we can easily get the following lemma, which will be useful in the next subsection.
\begin{lemma}
\label{lem_exp_f_inf}
Estimation of~$\langle f(\theta)\rangle_{M_\kappa}$.

Suppose that~$\theta\mapsto f(\theta) \, \sin^{n-2} \theta$ belongs to~$L^1(0,\pi)$, and that~$|f(\theta)|=O(\theta^{2\beta})$ in the neighborhood of~$0$, with~$\beta>-\frac{n-1}2$. 
Then~$\langle f(\theta)\rangle_{M_\kappa}=O(\kappa^{-\beta})$ as~$\kappa\to\infty$.
\end{lemma}
We now turn to the method to compute averages of the form~$\langle f(\theta) \rangle_{\widetilde M_{\kappa}}$.
\subsection{Asymptotics of~$\langle f(\theta) \rangle_{\widetilde M_{\kappa}}$}
We first decompose~$h_\kappa(\cos\theta)$ as a polynomial in~$\kappa$ or in~$\kappa^{-1}$ whose coefficients are polynomials of~$\cos\theta$ plus a remainder which will be negligible.
\begin{prop}
\label{expansion_h}
Expansion of~$h_\kappa$.

We define the two linear operators~$L$ and~$D$ on the space of polynomials by
\begin{gather*}
L(P)=-(1-X^2)P''+(n+1)XP'+(n-1)P \\
D(P)=-(1-X^2)P'+XP.
\end{gather*}
We have the two following expansions:
\begin{align*}
h_\kappa(\cos\theta)&=\sum_{p=0}^N H_p(\cos\theta)\kappa^p+ R^N_{\kappa,0}(\cos\theta),\\
h_\kappa(\cos\theta)&=\sum_{p=1}^N G_p^N(\cos\theta)\kappa^{-p}+ R^N_{\kappa,\infty}(\cos\theta),
\end{align*}
where the~$H_p$ (resp.~$G_p^N$) are the polynomials of degree~$p$ (resp. at most~$N-p$) given by the following induction relations (the second one being in the neighborhood of~$0$):
\begin{gather}
\begin{cases} L(H_0)=1\\ L(H_{p+1})=-D(H_p) \end{cases} \text{and}\quad
\begin{cases} D(G^N_{1})(\cos\theta)=1+O(\theta^{2N}) \\ (D(G^N_{p+1})+L(G^N_p))(\cos\theta)=O(\theta^{2(N-p)}),\end{cases}
\label{induction_Hp}\hspace{-2ex}
\end{gather}
and where the remainders satisfy the following estimations, for any function~$f$ such that~$\theta\mapsto f(\theta) \sin^{\frac{n}2} \theta$ belongs to~$L^2(0,\pi)$ and such that~$|f(\theta)|=O(\theta^{2\beta})$ in the neighborhood of~$0$:
\begin{align*}
\langle f(\theta)R^N_{\kappa,0}(\cos\theta)\sin^2\theta\rangle_{M_\kappa}&=O(\kappa^{N+1})\text{ as }\kappa\to0, \\
\langle f(\theta)R^N_{\kappa,\infty}(\cos\theta)\sin^2\theta\rangle_{M_\kappa}&=O(\kappa^{-\beta-N-2})\text{ as }\kappa\to\infty.
\end{align*}
\end{prop}

The proof of this proposition is given in Appendix~\ref{appendix_asymptotics}. 
The first thing to do is to prove that the inductions relations~\eqref{induction_Hp} make sense to define the sequence of polynomials (an induction relation is given in Appendix~\ref{appendix_asymptotics} to compute easily the polynomials~$G_p^N$ and~$H_p$).
The operators~$L$ and~$D$ are made so that
\begin{equation*}
\widetilde{L}_\kappa^*(P(\cos\theta)\sin\theta)=(L(P)+\kappa D(P))(\cos\theta)\sin\theta,
\end{equation*}
where the operator~$\widetilde{L}_\kappa^*$ is defined in~\eqref{Ltild_star_kappa}. Since we have~$\widetilde{L}_\kappa^*(h_\kappa(\cos\theta)\sin\theta)=\sin\theta$ by definition, we are then able to obtain the estimates on the remainders, using Poincaré inequalities in adapted spaces.

With this proposition, it is then easy to get an expansion for~$\langle f(\theta)\rangle_{\widetilde M_{\kappa}}$, using expressions of the form~$\langle g(\theta)\rangle_{M_\kappa}$, which can be expanded by the tools of the previous section:
\begin{equation*}
\langle f(\theta)\rangle_{\widetilde M_{\kappa}}=
\begin{cases}
\dfrac{\sum_{p=0}^N \langle f(\theta) H_p(\cos\theta)\sin^2\theta\rangle_{M_\kappa}\kappa^p}{\sum_{p=0}^N \langle H_p(\cos\theta)\sin^2\theta\rangle_{M_\kappa}\kappa^p}+ O(\kappa^{N+1}) &\text{as }\kappa\to0,\\
&\\
\dfrac{\sum_{p=1}^N \langle f(\theta) G_p^N(\cos\theta)\sin^2\theta\rangle_{M_\kappa}\kappa^{-p}}{\sum_{p=1}^N \langle G_p^N(\cos\theta)\sin^2\theta\rangle_{M_\kappa}\kappa^{-p}}+ O(\kappa^{-\beta-N})&\text{as }\kappa\to\infty.
\end{cases}
\end{equation*}

As an example, we can compute the first polynomials, we get
\begin{gather*}
H_0= \frac1{n-1},\quad  H_1=\frac{-X}{2n(n-1)},\quad G_1^2=\frac{4-X}3, \quad G_2^2=\frac{2(n-2)}3.
\end{gather*}
Hence, 
\begin{align*}
\langle\cos\theta\rangle_{\widetilde M_\kappa}&=\frac{\frac1{n-1}\langle\cos\theta\sin^2\theta\rangle_{M_\kappa}-\frac{\kappa}{2n(n-1)}\langle\cos^2\theta\sin^2\theta\rangle_{M_\kappa}}{\frac1{n-1}\langle\sin^2\theta\rangle_{M_\kappa}-\frac{\kappa}{2n(n-1)}\langle\cos\theta\sin^2\theta\rangle_{M_\kappa}}+O(\kappa^2),\\
\langle\cos\theta-1\rangle_{\widetilde M_\kappa}&=\frac{\frac1{3\kappa}\langle\cos\theta(4-\cos\theta)\sin^2\theta\rangle_{M_\kappa}+\frac{2(n-2)}{3\kappa^2}\langle\cos\theta\sin^2\theta\rangle_{M_\kappa}}{\frac1{3\kappa}\langle(4-\cos\theta)\sin^2\theta\rangle_{M_\kappa}+\frac{2(n-2)}{3\kappa^2}\langle\sin^2\theta\rangle_{M_\kappa}}-1+O(\kappa^{-3}).
\end{align*}
As before, in the second equation, we computed~$\langle\cos\theta-1\rangle_{\widetilde M_\kappa}$ instead of~$\langle\cos\theta\rangle_{\widetilde M_\kappa}$ in order to have a remainder of order~$3$ instead of~$2$.

Finally, we have to compute terms of the form~$\langle\cos^\ell\theta\sin^2\theta\rangle_{M_\kappa}$.
Instead of using the method of the previous subsection, we can actually express all these terms in function of~$c_1=\langle\cos\theta\rangle_{M_\kappa}$, by integrating by parts. 
We get
\begin{gather*}
\langle\sin^2\theta\rangle_{M_\kappa}=\tfrac{n-1}{\kappa}c_1,\quad\langle\cos\theta\sin^2\theta\rangle_{M_\kappa}=\tfrac{n-1}{\kappa}(1-\tfrac{n}{\kappa}c_1)\\
\langle\cos^2\theta\sin^2\theta\rangle_{M_\kappa}=\langle\sin^2\theta\rangle_{M_\kappa}-\langle\sin^4\theta\rangle_{M_\kappa}=\tfrac{n-1}{\kappa}(c_1-\tfrac{n+1}{\kappa}(1-\tfrac{n}{\kappa}c_1)).
\end{gather*}
Using the previous expansions~\eqref{exp_c1_zero} and~\eqref{exp_c1_inf}, we finally get the expansions of~$\widetilde{c_1}$:
\begin{equation}
\widetilde c_1=\langle\cos\theta\rangle_{\widetilde M_\kappa}=
\begin{cases}
\tfrac{2n-1}{2n(n+2)}\kappa+O(\kappa^2)&\text{as }\kappa\to0,\\
1-\frac{n+1}{2\kappa}+\frac{(n+1)(3n-7)}{24\kappa^2}+O(\kappa^{-3})&\text{as }\kappa\to\infty.
\end{cases}
\label{exp_c1tild}
\end{equation}
In addition, we can compute an expansion of~$\langle\sin^2\theta\rangle_{\widetilde M_\kappa}$, in order to get expansions for the coefficients~$c_2$ and~$\lambda$ of the macroscopic model, given in equations~\eqref{def_c2}-\eqref{def_lambda}.
Using only~$H_0$ and~$G_1^1=1$, we get that
\begin{align*}
\langle\sin^2\theta\rangle_{\widetilde M_\kappa}=
\begin{cases}
\dfrac{\langle\cos^2\theta\sin^2\theta\rangle_{M_\kappa}}{\langle\sin^2\theta\rangle_{M_\kappa}}+O(\kappa)=\frac{n+1}{n+2}+O(\kappa)&\text{as }\kappa\to0,\\
&\\
\dfrac{\langle\sin^4\theta\rangle_{M_\kappa}}{\langle\sin^2\theta\rangle_{M_\kappa}}+O(\kappa^{-2})=\frac{n+1}{\kappa}+O(\kappa^{-2})&\text{as }\kappa\to\infty.
\end{cases}
\end{align*}
So, using the expressions~\eqref{exp_c1_zero},~\eqref{exp_c1_inf} and~\eqref{exp_c1tild}, we get the following expansions:
\begin{align*}
c_1&=
\begin{cases}
\tfrac1n \kappa-\tfrac1{n^2(n+2)} \kappa^3+O(\kappa^5)&\text{as }\kappa\to0,\\
1- \frac{n-1}{2\kappa} +\frac{(n-1)(n-3)}{8\kappa^2} + O(\kappa^{-3})&\text{as }\kappa\to\infty,
\end{cases}\\
&\\
c_2&=
\begin{cases}
\tfrac{2n-1}{2n(n+2)}\kappa+O(\kappa^2)-\alpha\,\nu\,(\tfrac{2n+1}{2(n+2)}+O(\kappa))&\text{as }\kappa\to0,\\
1-\tfrac{n+1}2\,\kappa^{-1}+O(\kappa^{-2})-\alpha\,\nu\,(1-\kappa^{-1}+O(\kappa^{-2}))&\text{as }\kappa\to\infty,
\end{cases}\\
&\\
\lambda&=
\begin{cases}
\begin{array}{l}
\kappa^{-1}+\rho\,\tfrac{\dot{\kappa}}{\kappa}\,[ \,-\tfrac5{2n(n+2)}\kappa+O(\kappa^2) +\widetilde{\alpha}\,\nu\,(\tfrac{3}{2(n+2)}+O(\kappa)) \,] \\
\hspace{2ex}+ \,\tfrac12\,\widetilde{\alpha}\,\rho\,\tfrac{\dot{d}}{d}\,\nu\,((n-1)\kappa^{-1}+\tfrac{2n-1}{2n(n+2)}\kappa+O(\kappa^2))\,
\end{array}
&\text{as }\kappa\to0,
\\
\begin{array}{l}
\kappa^{-1}\,(1+\rho\,\tfrac{\dot{\kappa}}{\kappa}\,[ - 1 + O(\kappa^{-1})  +\widetilde{\alpha}\,\nu\,(1+O(\kappa^{-1})) \,]) \\
\hspace{2ex}+  \,\tfrac12\,\widetilde{\alpha}\,\rho\,\tfrac{\dot{d}}{d}\,\nu\,(1+\tfrac{n-3}{2}\kappa^{-1}+O(\kappa^{-2}))\,
\end{array}
&\text{as }\kappa\to\infty.
\end{cases}
\end{align*}

This shows that in any dimension, there are some simple cases were we actually have~$\lambda<0$ and the system loses hyperbolicity, even if~$\widetilde{\alpha}=0$ (for example if the kernel of observation~$\widetilde K$ is isotropic). 
For example if we have~$\kappa=\rho^\beta$ with~$\beta>0$, we have as~$rho\to\infty$ that~$\lambda(\rho)=(1-\beta)\rho^{-\beta}+O(\rho^{-2\beta})$, which gives that~$\lambda(\rho)<0$ if we take~$\beta>1$ and~$\rho$ sufficiently large.

These expansions also give a more precise estimation on the difference between~$c_1$ and~$c_2$ as the noise is small or large: when the kernel of observation is isotropic, we have~$c_1>c_2$ in the two expansions, in any dimension~$n$. 
That means that the information on the orientation propagates slower than the “fluid”.

\begin{remark}
We can also do an expansion in the more general case where~$\nu$ depends on~$\rho$ and~$\omega\cdot\Omega$. When~$d\to0$, the expansion of the coefficients depends only on the local behavior of the function~$x\mapsto\nu(\rho,x)$ near~$1$. 
In Appendix~\ref{tips} we give tips to perform this expansion. Here we only give the final expansion in the case where~$\nu$ and~$d$ do not depend on~$\rho$, so the coefficients are given by~\eqref{c_1_3d}-\eqref{lambda_3d}. In this case we can suppose~$\nu(1)=1$ (up to a rescaling), and denote~$\gamma=\nu'(1)$. We finally get, when~$n=2$:
\begin{align*}
 c_1&=1-\tfrac12 d+O(d^2),\\
c_2&= 1-\alpha + ((1+\tfrac32\gamma)\alpha-\tfrac32)d+O(d^2),\\
\lambda&=d+\tfrac32\gamma\, d^2+ O(d^3).
\end{align*}
\end{remark}

\section{Conclusion}

In this article, we have seen that the introduction of a dependence on the local density for two parameters at the microscopic level implies a significant change in the macroscopic limit: the possible loss of hyperbolicity in some regimes.
The introduction of a non-isotropic kernel of observation, without this dependence on the local density, is not sufficient to imply a strong difference of behavior for the continuum model. However, it enhances some properties, such as the difference between the velocity of the fluid and the velocity of the perturbations of the orientation.

It is important to note that the method introduced in~\cite{degond2008continuum} works to derive the macroscopic model. In particular the concept of generalized collisional invariants is still valid, with some adaptations, and we get the same macroscopic model, except for the definition of the coefficients.

Some questions are left open. The limit here is formal, and we are still looking for an appropriate functional framework to obtain more precise results of convergence. The rigorous derivation of the mean-field limit of the dynamical system of particles is also part of our future work.

Finally, the next step to this study consists in numerical simulations, in order to see how the difference between~$c_2$ and~$c_1$ can be observed in simulations of the discrete dynamical system, or how the particles behave in the regions of non-hyperbolicity.

\appendix
\section{Proof of some statements for section~\ref{elements}}

\subsection{Expansion of the local density and orientation}
\label{annex_lemma_expansion}

We recall the expressions of~$\bar{\omega}^\varepsilon$ and~$\bar{\rho}^\varepsilon$:

\begin{align}
\bar{\omega}^\varepsilon (x, \omega, t) &= \frac{\bar J^\varepsilon(x,\omega,t)}{|\bar J^\varepsilon(x,\omega,t)|}, \label{bar_Omega_eps}\\
\bar J^\varepsilon(x,\omega,t) &= \int_{ y \in \mathbb{R}^n , \, \upsilon \in \mathbb{S} } K\left(\tfrac{|x-y|}{\varepsilon},\tfrac{y-x}{|x-y|}\cdot\omega\right) \, 
\upsilon \, f^\varepsilon(y, \upsilon,t) \, \tfrac{\mathrm d y}{\varepsilon^n} \, \mathrm d\upsilon \, , \label{Jeps}\\
\bar{\rho}^\varepsilon (x, \omega, t) &= \int_{ y \in \mathbb{R}^n , \, \upsilon \in \mathbb{S} } \widetilde K\left(\tfrac{|x-y|}{\varepsilon},\tfrac{y-x}{|x-y|}\cdot\omega\right) \, f^\varepsilon(y, \upsilon,t) \, \tfrac{\mathrm d y}{\varepsilon^n}\, \mathrm d\upsilon \, , \label{rho_eps}
\end{align}

\begin{lemma} We have the following expansions: 
\label{lemma_expansion}
\begin{align*} 
\bar{\omega}^\varepsilon (x, \omega, t) &= \Omega^\varepsilon (x, t) + \varepsilon \alpha\, (\omega\cdot\nabla_x)\, \Omega^\varepsilon(x,t) + O(\varepsilon^2) \, ,\\
\bar{\rho}^\varepsilon (x, \omega, t) &= \rho^\varepsilon (x, t) + \varepsilon \widetilde{\alpha}\, \omega\cdot\nabla_x\rho^\varepsilon(x,t) + O(\varepsilon^2) \, .
\end{align*}
where the constants~$\alpha$ and~$\widetilde{\alpha}$ depend only on the observation kernels~$K$ and~$\widetilde{K}$, and
\begin{align*} 
\Omega^\varepsilon(x,t) &= \tfrac{j^\varepsilon(x,t)}{|j^\varepsilon (x,t)|}, \text{ with } j^\varepsilon (x,t) = \int_{\upsilon \in \mathbb{S}} \upsilon\,f^\varepsilon(x,\upsilon,t)\, \mathrm d\upsilon\, ,\\
\rho^\varepsilon(x,t) &= \int_{\upsilon\in\mathbb{S}} f^\varepsilon (x, \upsilon,t) \, \mathrm d\upsilon \,.
\end{align*}
\end{lemma}

\begin{proof}
After change of variable~$y=x+\varepsilon\xi$, let us expand~$f$ at first order in~$\varepsilon$ in~\eqref{Jeps}. We get
\begin{equation*}
\bar J^\varepsilon(x,\omega,t)=\int_{\xi \in \mathbb{R}^n , \, \upsilon \in \mathbb{S}}\hspace{-1cm} K (|\xi |,\tfrac{ \xi}{|\xi |}\cdot\omega) \, \upsilon \,(f^\varepsilon(x, \upsilon,t)+\varepsilon \, \xi\cdot \nabla_x f^\varepsilon(x,\upsilon,t) + O(\varepsilon^2))\, \mathrm d\xi \, \mathrm d\upsilon \, .
\end{equation*}
We have to compute
\begin{equation*}
K_0(\omega)=\int_{ \xi \in \mathbb{R}^n} K (|\xi |,\tfrac{ \xi}{|\xi |}\cdot\omega) \mathrm d\xi \quad \text{ and } \quad 
K_1(\omega)=\int_{ \xi \in \mathbb{R}^n} K (|\xi |,\tfrac{ \xi}{|\xi |}\cdot\omega) \, \xi \, \mathrm d\xi.
\end{equation*}
For any rotation~$R$, the change of variable~$\tilde{\xi}=R(\xi)$ gives on one hand 
\begin{equation*}
K_0(\omega)=K_0(R(\omega))\, ,
\end{equation*} 
and so~$K_0$ does not depend on~$\omega$. 
On the other hand, we get
\begin{equation*}
R(K_1(\omega))=K_1(R(\omega)) \, ,
\end{equation*}
which shows that~$K_1(\omega)$ is a vector invariant by any rotation which let~$\omega$ invariant, so it is parallel to~$\omega$. Given a vector~$e$ of~$\mathbb{S}$, we have~$K_1(e)=k_1 e$. Then taking one rotation mapping~$\omega$ to~$e$, we get~$R(K_1(\omega))=K_1(e)=k_1 e = R(k_1 \omega)$, so finally we get~$K_1(\omega)=k_1 \omega$ for all~$\omega\in\mathbb{S}$.
Let then~$\alpha=\frac{k_1}{K_0}$, and we have 
\begin{align*}
\frac{\bar J^\varepsilon(x,\omega,t)}{K_0} & = \int_{\upsilon \in \mathbb{S}}
\upsilon \, (f^\varepsilon(x, \upsilon,t)+\varepsilon \, \alpha \, \omega\cdot \nabla_x f^\varepsilon(x, \upsilon,t)) \, \mathrm d\upsilon + O(\varepsilon^2)\\
&= j^\varepsilon(x,t) + \varepsilon \, \alpha\, (\omega\cdot \nabla_x) j^\varepsilon(x,t) + O(\varepsilon^2) \, .
\end{align*}
Putting this expression into~\eqref{bar_Omega_eps}, we get
\begin{gather*}
\left|\frac{\bar J^\varepsilon(x,\omega,t)}{K_0}\right|^2 = |j^\varepsilon(x,t)|^2 + 2\, \varepsilon \, \alpha\,j^\varepsilon(x,t)\cdot(\omega\cdot \nabla_x) j^\varepsilon(x,t) + O(\varepsilon^2) \, , 
\intertext{so}
\left|\frac{\bar J^\varepsilon(x,\omega,t)}{K_0}\right|^{-1} = \frac1{|j^\varepsilon(x,t)|}\left(1 - \frac{\varepsilon \, \alpha}{|j^\varepsilon(x,t)|^2}\,j^\varepsilon(x,t)\cdot(\omega\cdot \nabla_x) j^\varepsilon(x,t)\right) + O(\varepsilon^2) \, ,
\intertext{and finally}
\begin{split}
\bar{\omega}^\varepsilon (x, \omega, t) &= \frac{j^\varepsilon(x,t)}{|j^\varepsilon(x,t)|} \\
& + \varepsilon \, \alpha\, \left(\frac{(\omega\cdot\nabla_x)j^\varepsilon(x,t)}{|j^\varepsilon(x,t)|}-\frac{j^\varepsilon(x,t)}{|j^\varepsilon(x,t)|}\cdot\frac{(\omega\cdot \nabla_x)j^\varepsilon(x,t)}{|j^\varepsilon(x,t)|} \frac{j^\varepsilon(x,t)}{|j^\varepsilon(x,t)|}\right) \, + O(\varepsilon^2).
\end{split}
\end{gather*}
But we also have
\begin{align*}
(\omega\cdot\nabla_x)\Omega^\varepsilon(x,t)&=\frac{(\omega\cdot \nabla_x)j^\varepsilon(x,t)}{|j^\varepsilon(x,t)|}+ \left(\omega\cdot \nabla_x\left(\frac{1}{|j^\varepsilon(x,t)|}\right)\right)\, j^\varepsilon(x,t) \\
& = \frac{(\omega\cdot\nabla_x)j^\varepsilon(x,t)}{|j^\varepsilon(x,t)|}-\frac{1}{|j^\varepsilon(x,t)|^3}\, (((\omega\cdot\nabla_x)j^\varepsilon(x,t))\cdot j^\varepsilon(x,t))\, j^\varepsilon(x,t)\, .
\end{align*}
Therefore
\begin{equation*}
\bar{\omega}^\varepsilon (x, \omega, t) = \Omega^\varepsilon (x, t) + \varepsilon \alpha\, (\omega\cdot\nabla_x)\, \Omega^\varepsilon(x,t) + O(\varepsilon^2) \, ,
\end{equation*}
and this is the first part of the lemma.

After the same change of variable~$y=x+\varepsilon\xi$ and expansion in~\eqref{rho_eps}, and using the same techniques, and the normalization condition~\eqref{rho_bar_normalization}, we get
\begin{align*}
\bar{\rho}^\varepsilon(x,\omega,t) & = \int_{\upsilon \in \mathbb{S}} f^\varepsilon(x, \upsilon,t)+\varepsilon \, \widetilde K_1(\omega)\cdot \nabla_x f^\varepsilon(x, \upsilon,t) \, \mathrm d\upsilon + O(\varepsilon^2)\\
&= \rho^\varepsilon(x,t) + \varepsilon \, \widetilde{\alpha}\, \omega\cdot \nabla_x\rho^\varepsilon(x,t) + O(\varepsilon^2) \, .
\end{align*}
This is the second part of the lemma.
\end{proof}

\subsection{Proof of Proposition~\ref{prop_dtOmega}}
\label{annex_dtOmega}
We have to compute~$(\mathrm{Id} - \Omega \otimes \Omega)\, X$, where 
\begin{equation*} 
X = \int_{\omega \in \mathbb{S}} ( (\partial_t+ \omega \cdot \nabla_x) (\rho M_{\kappa\Omega}) +\alpha P(\rho M_{\kappa\Omega})+\widetilde{\alpha} \widetilde{P}(\rho M_{\kappa\Omega})) \, h_\kappa(\omega \cdot \Omega) \, \omega \, \mathrm d\omega\, .
\end{equation*}
For convenience, we will write~$\nu$,~$d$ for~$\nu(\rho)$,~$d(\rho)$ in the following.
We first give some useful formulas to work on the unit sphere. For~$V$ a constant vector in~$\mathbb{R}^n$, we have:
\begin{gather*}
 \nabla_\omega (\omega\cdot V) = (\mathrm{Id} - \omega \otimes \omega) V,\\
 \nabla_\omega\cdot((\mathrm{Id} - \omega \otimes \omega)V)= -(n-1)\,\omega\cdot V.
\end{gather*}
Then we have that for any constant matrix~$A$ 
\begin{equation*}
 \nabla_\omega\cdot((\mathrm{Id} - \omega \otimes \omega)A\omega)= A:(\mathrm{Id}-n\omega\otimes\omega),
\end{equation*}
where the notation~“$:$” denotes the “contraction” of two operators (if~$A = (A_{ij})$ and~$B = (B_{ij})$ then~$A:B = \sum_{i,j=1, \ldots, n} A_{ij} B_{ij}$, this is the trace of~$AB^T$). This can be shown when~$A$ is of the form~$V_1\otimes V_2$, using the previous formulas, and then extended by linearity.

We recall the definition of~$M_{\kappa\Omega}$, given in equation~\eqref{M_def}:
\begin{equation*}
M_{\kappa\Omega}(\omega) =\frac{e^{\kappa\, \omega \cdot \Omega}}{\int_\mathbb{S} e^{\kappa\, \upsilon \cdot \Omega}\, \mathrm d\upsilon}.
\end{equation*}
We get, writing~$\cos\theta$ for~$\omega\cdot\Omega$, and using the notation~$\langle\cdot\rangle_{M_\kappa}$ given in~\eqref{brackets},
\begin{align*}
\nabla_\omega M_{\kappa\Omega}&= \kappa(\mathrm{Id}-\omega\otimes\omega)\Omega M_{\kappa\Omega},\\
\nabla_\Omega M_{\kappa\Omega}&= \kappa(\mathrm{Id}-\Omega\otimes\Omega)\omega M_{\kappa\Omega},\\
\partial_\kappa M_{\kappa\Omega}&= (\cos\theta-\langle\cos\theta\rangle_{M_\kappa})M_{\kappa\Omega}.
\end{align*}
Using the chain rule, we then get
\begin{equation*}
\begin{split}
(\partial_t+ \omega \cdot \nabla_x) (\rho M_{\kappa\Omega})=(1&+(\cos\theta-\langle\cos\theta\rangle_{M_\kappa})\rho\dot{\kappa})M_{\kappa\Omega}(\partial_t+ \omega \cdot \nabla_x)\rho \\
&+ \rho\kappa(\mathrm{Id}-\Omega\otimes\Omega)\omega M_{\kappa\Omega}\cdot(\partial_t+ \omega \cdot \nabla_x)\Omega,
\end{split}
\end{equation*}
where~$\dot{\kappa}$ is the derivative of~$\kappa$ with respect to~$\rho$.
Since~$\Omega$ is of norm~$1$, we have that~$(\partial_t+ \omega \cdot \nabla_x)\Omega$ is orthogonal to~$\Omega$, and the term~$\Omega\otimes\Omega$ vanishes.
We get 
\begin{equation*}
\begin{split}
(\partial_t+ \omega \cdot \nabla_x) (\rho M_{\kappa\Omega})=(1&+(\cos\theta-\langle\cos\theta\rangle_{M_\kappa})\rho\dot{\kappa})M_{\kappa\Omega}(\partial_t\rho+ \omega \cdot \nabla_x\rho) \\
&+ \rho\kappa M_{\kappa\Omega}(\omega\cdot\partial_t\Omega+ \omega\otimes\omega: \nabla_x\Omega),
\end{split}
\end{equation*}
where~$\nabla_x \Omega$ is the gradient tensor of~$\Omega$ that is to say~$(\nabla_x \Omega)_{ij} = \partial_{x_i} \Omega_j$.
We then have
\begin{align*}
P(\rho&M_{\kappa\Omega})  = \nu(\rho) \nabla_\omega \cdot ((\mathrm{Id} - \omega \otimes \omega)((\omega\cdot\nabla_x)\, \Omega) \rho M_{\kappa\Omega}),\\
& = \rho  \nu(\rho) [\kappa \Omega\cdot(\mathrm{Id} - \omega \otimes \omega)((\nabla_x\Omega)^T\omega) + \nabla_\omega \cdot ((\mathrm{Id} - \omega \otimes \omega)(\nabla_x\Omega)^T\omega)] M_{\kappa\Omega},
\end{align*}
where the notation~$^T$ denotes the transpose of operators. Hence, using the fact that~$(\nabla_x\Omega)^T\omega=(\omega\cdot\nabla_x)\, \Omega$ is orthogonal to~$\Omega$, and the formula given in the beginning of this section, with~$A=(\nabla_x\Omega)^T$, and  we get
\begin{align*}
P(\rho M_{\kappa\Omega}) & = \rho  \nu(\rho) [-\kappa \cos\theta\, \omega\otimes\omega: (\nabla_x\Omega)^T + (\nabla_x\Omega)^T:(\mathrm{Id} - n \omega \otimes \omega)] M_{\kappa\Omega},\\
&=\rho  \nu [\nabla_x\cdot\Omega-(n+\kappa \cos\theta)\, \omega\otimes\omega: \nabla_x\Omega] M_{\kappa\Omega}.
\end{align*}
Similarly, for the operator~$\widetilde{P}$, we get
\begin{align*}
\widetilde{P}(\rho M_{\kappa\Omega}) &=\dot{\nu}(\rho) \nabla_\omega \cdot ((\omega\cdot\nabla_x\rho)\,(\mathrm{Id}-\omega\otimes\omega)\Omega \rho M_{\kappa\Omega})\\
  &\hspace{1cm}- \dot d(\rho) \nabla_\omega \cdot (\tfrac12\rho M_{\kappa\Omega}(\mathrm{Id}-\omega\otimes\omega)\nabla_x\rho + (\omega\cdot\nabla_x\rho)\, \nabla_\omega \rho M_{\kappa\Omega}),\\
&=\rho(\dot{\nu}-\kappa\dot d)\nabla_\omega \cdot ((\mathrm{Id}-\omega\otimes\omega)(\Omega\otimes\nabla_x\rho)\omega M_{\kappa\Omega}) \\
  &\hspace{1cm}- \tfrac12\rho\dot d [\nabla_\omega\cdot ((\mathrm{Id}-\omega\otimes\omega)\nabla_x\rho)+\kappa\Omega\cdot(\mathrm{Id}-\omega\otimes\omega)\nabla_x\rho] M_{\kappa\Omega}.
\end{align*}
But we have~$\nu=\kappa d$, so~$\dot{\nu}-\kappa\dot d=d\dot{\kappa}$. And we have
\begin{align*}
\nabla_\omega \cdot ((\mathrm{Id}&-\omega\otimes\omega)(\Omega\otimes\nabla_x\rho)\omega M_{\kappa\Omega})\\
&=\nabla_\omega \cdot ((\mathrm{Id}-\omega\otimes\omega)(\Omega\otimes\nabla_x\rho)\omega)M_{\kappa\Omega} +\kappa\Omega\cdot(\mathrm{Id}-\omega\otimes\omega)(\Omega\otimes\nabla_x\rho)\omega]\\
&=[(\Omega\otimes\nabla_x\rho):(\mathrm{Id}-n\omega\otimes\omega)+\kappa(1-\cos^2\theta)\, \omega\cdot\nabla_x\rho]M_{\kappa\Omega}\\
&=[\Omega\cdot\nabla_x\rho +(\kappa\sin^2\theta-n\cos\theta)\, \omega\cdot\nabla_x\rho]M_{\kappa\Omega}.
\end{align*}
Hence
\begin{align*}
\widetilde{P}(\rho M_{\kappa\Omega}) &=\rho d\dot{\kappa}[\Omega\cdot\nabla_x\rho +(\kappa\sin^2\theta-n\cos\theta)\, \omega\cdot\nabla_x\rho]M_{\kappa\Omega} \\
  &\hspace{1cm}+ \tfrac12\rho\dot d [(\kappa\cos\theta+n-1)\, \omega\cdot\nabla_x\rho - \kappa \Omega\cdot\nabla_x\rho] M_{\kappa\Omega}.
\end{align*}
Finally we can write~$X=X_1+X_2+X_3$, where
\begin{align*}
X_1&=\int_{\omega \in \mathbb{S}}h_\kappa(\cos\theta)\gamma_1(\cos\theta)\,\omega\,M_{\kappa\Omega}\mathrm d \omega, \\
X_2&=\int_{\omega \in \mathbb{S}}h_\kappa(\cos\theta)\,\omega\otimes\omega (\gamma_2(\cos\theta)\,\nabla_x\rho+ \rho\kappa \partial_t\Omega)\,M_{\kappa\Omega}\mathrm d \omega, \\
X_3&=\int_{\omega \in \mathbb{S}}h_\kappa(\cos\theta)\gamma_3(\cos\theta)\,\omega(\omega\otimes\omega:\nabla_x\Omega)\,M_{\kappa\Omega}\mathrm d \omega,
\end{align*}
with (using the notation~$c_1=\langle\cos\theta\rangle_{M_\kappa}$)
\begin{align*}
\gamma_1(\cos\theta)&=(1+(\cos\theta-c_1)\rho\dot{\kappa})\partial_t\rho+ \alpha\rho \nu\nabla_x\cdot\Omega + \widetilde{\alpha}\rho(d\dot{\kappa}-\tfrac12\dot d\kappa)\Omega\cdot\nabla_x\rho,\\
\gamma_2(\cos\theta)&=1+(\cos\theta-c_1)\rho\dot{\kappa}+\widetilde{\alpha}\rho(d\dot{\kappa}(\kappa\sin^2\theta-n\cos\theta)+\tfrac12\dot d\kappa(\kappa\cos\theta+n-1)),\\
\gamma_3(\cos\theta)&=\rho\,\kappa -\alpha\rho\nu(n+\kappa \cos\theta).
\end{align*}

To do the computation we write~$\omega=\cos\theta\, \Omega+\sin\theta\, v$, with~$v\in\mathbb{S}_{n-2}$ (identified with the set of unit vectors which are orthogonal to~$\Omega$).
We take the following convention:~$\int_{v \in \mathbb{S}_{n-2}}\mathrm d v=1$, and we have
\begin{gather*}
\int_{\omega \in \mathbb{S}_{n-1}}\hspace{-0.5cm}a(\omega)\mathrm d \omega=\frac{1}{V_n}\int_0^\pi\int_{v \in \mathbb{S}_{n-2}}\hspace{-0.5cm}a(\theta,v)\sin^{n-2}\theta\,\mathrm d v\,\mathrm d \theta,\\
\int_{v \in \mathbb{S}_{n-2}}\hspace{-0.5cm}v\,\mathrm d v=0, \text{ and } \int_{v \in \mathbb{S}_{n-2}}\hspace{-0.5cm}v\otimes v\,\mathrm d v=\frac{1}{n-1}(\mathrm{Id}-\Omega\otimes\Omega),
\end{gather*}
where~$V_n$ is a normalization constant (we will not need it in the following). These results are still valid when~$n=2$, with~$\int_{v \in \mathbb{S}_{0}}\mathrm d v=\frac12(a(v_0)+a(-v_0))$,~$v_0$ being one of the two unit vectors orthogonal to~$\Omega$. Using these formulas, we get
\begin{gather*}
\int_{\omega \in \mathbb{S}}\gamma(\cos\theta) \, M_{\kappa\Omega} \, \omega\, \mathrm d\omega = \langle\cos \theta \, \gamma(\cos\theta)\rangle_{M_\kappa} \, \Omega,\\
\int_{\omega \in \mathbb{S}}\omega\otimes\omega\, \gamma(\cos\theta)\, M_{\kappa\Omega}\, \mathrm d\omega = \langle \cos^2 \theta \, \gamma \rangle_{M_\kappa} \Omega \otimes \Omega +\frac{\langle \sin^2 \theta \, \gamma \rangle_{M_\kappa}}{n-1} (\mathrm{Id} - \Omega \otimes \Omega). 
\end{gather*}
So we have (knowing that~$\partial_t\Omega$ is orthogonal to~$\Omega$):
\begin{gather*}
 (\mathrm{Id} - \Omega \otimes \Omega)\, X_1=0, \\
 (\mathrm{Id}-\Omega\otimes\Omega)\,X_2= \frac{\langle\sin^2 \theta\,\gamma_2h_\kappa\rangle_{M_\kappa}}{n-1} (\mathrm{Id} - \Omega \otimes \Omega)\nabla_x \rho +\frac{\rho \kappa\langle\sin^2 \theta\, h_\kappa\rangle_{M_\kappa}}{n-1}\partial_t\Omega. 
\end{gather*}
To compute~$(\mathrm{Id}-\Omega\otimes\Omega)\,X_3$, we first remark that
\begin{equation*}
(\mathrm{Id}-\Omega\otimes\Omega)\,\omega(\omega\otimes\omega:\nabla_x\Omega)=\sin\theta\,v(\omega\cdot(\omega\cdot\nabla_x)\Omega)=\sin^2\theta\,v(v\cdot(\omega\cdot\nabla_x)\Omega),
\end{equation*}
since~$(\omega\cdot\nabla_x)\Omega$ is orthogonal to~$\Omega$. But we have~$\int_{v \in \mathbb{S}_{n-2}}v(v\otimes v\!:\!\nabla_x\Omega)\,\mathrm d v=0$, because the integrand is odd with respect to~$v$, and then we get 
\begin{align*}
(\mathrm{Id}-\Omega\otimes\Omega)\,X_3&=\langle \sin^2 \theta \,\cos\theta\, \gamma_3\,h_\kappa \rangle_{M_\kappa}\int_{v \in \mathbb{S}_{n-2}}\hspace{-0.5cm}v\otimes v\,\mathrm d v\,(\Omega\cdot\nabla_x)\Omega\\
&=\frac{\langle \sin^2 \theta \,\cos\theta\, \gamma_3\,h_\kappa \rangle_{M_\kappa}}{n-1}(\Omega\cdot\nabla_x)\Omega,
\end{align*}
since~$(\Omega\cdot\nabla_x)\Omega$ is orthogonal to~$\Omega$.

So we have that~$(\mathrm{Id} - \Omega \otimes \Omega) X=0$ is equivalent to 
\begin{equation*}
\rho \kappa\langle\sin^2 \theta\, h_\kappa\rangle_{M_\kappa}\partial_t\Omega+\langle \sin^2 \theta \,\cos\theta\, \gamma_3\,h_\kappa \rangle_{M_\kappa}(\Omega\cdot\nabla_x)\Omega+\langle\sin^2 \theta\,\gamma_2h_\kappa\rangle_{M_\kappa}\nabla_x \rho=0.
\end{equation*}

For any function~$\gamma(\cos \theta)$, we denote by~$\langle \gamma(\cos \theta) \rangle_{\widetilde M_\kappa}$ the mean of~$\gamma(\cos \theta)$ following the “weight”~$\sin^2 \theta\,h_\kappa(\cos \theta)  M_{\kappa\Omega}$, that is to say
\begin{equation*}
\langle\gamma(\cos \theta)\rangle_{\widetilde M_\kappa}= \frac{\int_0^\pi\gamma(\cos \theta)h_\kappa(\cos\theta) e^{\kappa\cos\theta}\,\sin^n\theta\,\mathrm d\theta}{\int_0^\pi h_\kappa(\cos\theta)e^{\kappa\cos\theta}\,\sin^n\theta\,\mathrm d\theta}.
\end{equation*}
We have
\begin{equation*}
\langle \gamma(\cos \theta) \rangle_{\widetilde M_\kappa}=\frac{\langle \sin^2 \theta\,h_\kappa(\cos\theta)\gamma(\cos \theta) \rangle_{M_\kappa}}{\langle \sin^2 \theta\,h_\kappa(\cos\theta)\rangle_{M_\kappa}},
\end{equation*}
and so, dividing by~$\kappa\langle \sin^2 \theta h_\kappa(\cos\theta)\rangle_{M_\kappa}$ we finally get that~$(\mathrm{Id} - \Omega \otimes \Omega) X=0$ is equivalent to 
\begin{equation*}
\rho  \left( \partial_t \Omega + c_2 (\Omega \cdot \nabla_x) \Omega \right) + \lambda \, (\mathrm{Id} - \Omega \otimes \Omega) \nabla_x \rho = 0,
\end{equation*}
where the coefficients are given by
\begin{equation*}
c_2(\rho)=\tfrac1{\kappa\rho}\langle \cos \theta \,\gamma_3(\cos\theta) \rangle_{\widetilde M_\kappa}\text{ and }\lambda(\rho)=\tfrac1{\kappa}\langle\gamma_2(\cos\theta) \rangle_{\widetilde M_\kappa}.
\end{equation*}
We finally get, writing~$\widetilde c_1=\langle \cos\theta \rangle_{\widetilde M_\kappa}$,
\begin{align*}
c_2&=\widetilde c_1 - \alpha\,d\,(n\, \widetilde c_1+ \kappa\,\langle\cos^2\theta\rangle_{\widetilde M_\kappa})\, ,\\
\lambda&=\tfrac1{\kappa}+\rho\,\tfrac{\dot{\kappa}}{\kappa}\,[\,(\,\widetilde c_1 -c_1+\widetilde{\alpha}\,d\, (\kappa\langle \sin^2\theta\rangle_{\widetilde M_\kappa}-n\,\widetilde c_1)\,] + \tfrac12\,\widetilde{\alpha}\,\dot{d}\,(\kappa\,\widetilde c_1 +n-1)\, ,
\end{align*}
which are exactly the expressions given in equations~\eqref{def_c2}-\eqref{def_lambda}, and this ends the proof of Proposition~\ref{prop_dtOmega}.

\section{Asymptotics of the coefficients}
 \subsection{Proof of Proposition~\ref{expansion_h}}
\label{appendix_asymptotics}

We recall that the two linear operators~$L$ and~$D$ on the space of polynomials are defined by
\begin{gather*}
L(P)=-(1-X^2)P''+(n+1)XP'+(n-1)P \\
D(P)=-(1-X^2)P'+XP.
\end{gather*}
We first give a preliminary lemma which will be helpful to construct the polynomials~$H_p$ and~$G_p^N$.
\begin{lemma}Definition of the polynomials.
\label{lem_def_polynomials}

Let~$Q$ be a polynomial and~$N\in\mathbb{N}$. Then
\begin{itemize}
\item There exists one unique polynomial~$P$ such that~$L(P)=Q$.
\item There exists one unique polynomial~$P^N$ of degree at most~$N$ such that
\begin{equation*}
D(P^N)(\cos\theta)=Q(\cos\theta)+O(\theta^{2(N+1)})\text{ as }\theta\to\theta.
\end{equation*}
\end{itemize}
\end{lemma}
\begin{proof}
For the first point, if the leading term in a polynomial~$P$ is~$a_kX^k$, with~$a_k\neq0$, then the leading term in~$L(P)$ is~$[k(k-1)+k(n+1)+(n-1)]a_kX^k$, and so~$L(P)\neq0$. 
So the linear operator~$L$ is injective from~$\mathbb{R}_p[X]$ to~$\mathbb{R}_p[X]$, and therefore it is bijective.

For the second point, the idea is to remark that
\begin{equation*}
D((1-X)^k)=(2k+1)(1-X)^k+(k+1)(1-X)^{k+1},
\end{equation*}
so we write the polynomials in the basis~$\{(1-X)^k,k\in\mathbb{N}\}$. 
We get that a polynomial~$R$ is such that~$R(\cos\theta)=O(\theta^{2(N+1)})$ if and only if, in this basis, its first coefficients up to order~$(X-1)^N$ are zero (because~$1-\cos\theta=\tfrac12\theta^2 + O(\theta^4)$ in the neighborhood of~$0$).
We write~$Q=\sum_{k=0}^\infty b_k(1-X)^k$ and~$P^N=\sum_{k=0}^Na_k(1-X)^k$, and we get that
\begin{equation*}
  D(P^N)(\cos\theta)=Q(\cos\theta)+O(\theta^{2(N+1)})\Leftrightarrow
\begin{cases}
a_0=b_0&\\
(2k+1) a_k-ka_{k-1}=b_k& \forall k\in\llbracket1,N\rrbracket.
\end{cases}
\end{equation*}
Since this induction relation defines in an unique way the coefficients~$a_k$ for~$k\in\llbracket1,N\rrbracket$, this ends the proof.
\end{proof}

With this lemma, we can now define the following sequences of polynomials~$H_p$ and~$G_p^N$, this last ones being of degree at most~$N-p$:
\begin{gather*}
\begin{cases} L(H_0)=1\\ L(H_{p+1})=-D(H_p) \end{cases} \text{and}\quad
\begin{cases} D(G^N_{1})(\cos\theta)=1+O(\theta^{2N}) \\ (D(G^N_{p+1})+L(G^N_p))(\cos\theta)=O(\theta^{2(N-p)}).\end{cases}
\end{gather*}
Since the operator~$L$ is odd and~$D$ is even, it is easy to show that the polynomials~$H_p$ have the same parity as~$p$. If we express the operator~$L$ in the basis~$\{(1-X)^k,k\in\mathbb{N}\}$, we are able to get the induction relation for the coefficients of the polynomials in this basis. We have
\begin{equation*}
L((1-X)^k)=(n+k-1)(k+1)(1-X)^k-k(n+2k-1)(1-X)^{k-1},
\end{equation*}
So we have
\begin{equation}
\label{expression_H_G}
H_p=\sum_{k=0}^pb^p_k(1-X)^k, \text{ and } G^N_p=\sum_{k=0}^{N-p}a^p_k(1-X)^k,
\end{equation}
where~$a_k^p$ and~$b_k^p$ are given by the following induction relations for (with the convention that~$b_{p+1}^p=b_{-1}^p=a_{-1}^p=0$):
\begin{align*}
&\begin{cases}
b^0_0=\frac{1}{n-1},&\\
(n+k-1)b^{p+1}_k-(n+2k+1)b^{p+1}_{k+1}=\frac{2k+1}{k+1} b^p_k-\frac{k}{k+1}b^p_{k-1},& \forall p\in\mathbb{N},\forall k=\llbracket0,p+1\rrbracket.
\end{cases}\\
&\begin{cases}
a^0_k=\frac{k!}{(2k+1)(2k-1)\dots3}=\frac{2^k(k!)^2}{(2k+1)!},&\\
\frac{2k+1}{k+1} a^{p+1}_k-\frac{k}{k+1}a^{p+1}_{k-1}=(n+k-1)a^p_k-(n+2k+1)a^p_{k+1},& \forall p\in\mathbb{N},\forall k\in\mathbb{N}.
\end{cases}
\end{align*}

We define then the remainders~$R^N_{\kappa,0}$ and~$R^N_{\kappa,\infty}$ by
\begin{equation*}
R^N_{\kappa,0}(\mu)=h_\kappa(\mu)-\sum_{p=0}^N H_p(\mu)\kappa^p,\quad R^N_{\kappa,\infty}(\mu)=h_\kappa(\mu)-\sum_{p=1}^N G_p^N(\mu)\kappa^{-p}.
\end{equation*}
It is an easy matter to see that, for a given polynomial~$P$, we have
\begin{gather}
\sin\theta\partial_\theta(P(\cos\theta)\sin\theta)=D(P)(\cos\theta)\sin\theta,\label{D_theta}\\
-\partial_\theta(\sin^{n-2}\theta\partial_\theta(P(\cos\theta)\sin\theta))+(n-2)\sin^{n-3}\theta P(\cos\theta)=\sin^{n-1}\theta L(P)(\cos\theta),\nonumber
\end{gather}
and then we get
\begin{equation*}
\widetilde L_\kappa^*(P(\cos\theta)\sin\theta)=(L(P)+\kappa D(P))(\cos\theta)\sin\theta,
\end{equation*}
where the operator~$\widetilde{L}_\kappa^*$ is defined in~\eqref{Ltild_star_kappa} by
\begin{equation*}
\widetilde L_\kappa^*g(\theta)=-\sin^{2-n}\theta e^{-\kappa\cos\theta}\tfrac{\mathrm d}{\mathrm d\theta}(\sin^{n-2}\theta e^{\kappa\cos\theta}g'(\theta))+\tfrac{n-2}{\sin^2\theta}g(\theta).
\end{equation*}
Since we have by definition~$\widetilde{L}_\kappa^*(h_\kappa(\cos\theta)\sin\theta)=\sin\theta$, we get
\begin{align*}
\widetilde{L}_\kappa^*(R^N_{\kappa,0}(\cos\theta)\sin\theta)&=\sin\theta-\sum_{p=0}^N (L(H_p)+\kappa D(H_p))(\cos\theta)\sin\theta\kappa^p\\
&=-\kappa^{N+1}D(H_N)(\cos\theta)\sin\theta,\\
\widetilde{L}_\kappa^*(R^N_{\kappa,\infty}(\cos\theta)\sin\theta)&=\sin\theta-\sum_{p=1}^N (L(G^N_p)+\kappa D(G^N_p))(\cos\theta)\sin\theta\kappa^{-p}\\
&=-L(G^N_N)(\cos\theta)\sin\theta\kappa^{-N}+\sum_{p=0}^{N-1}\kappa^{-p}O(\theta^{2(N-p)})\sin\theta.
\end{align*}
To get estimations for the averages of the form~$\langle f(\theta)R^N_{\kappa,\varepsilon}(\cos\theta)\sin^2\theta\rangle_{M_\kappa}$ (with~$\varepsilon$ standing for~$0$ or~$\infty$) we first remark that, for a function~$g$ belonging to the space~$V$ (a “weighted~$H^1_0$”) defined in~\eqref{def_V} by
\begin{equation*}
V = \{ g \, | \,(n-2)(\sin\theta)^{\frac n2-2} g \in L^2(0,\pi), \, (\sin\theta)^{\frac n2-1}g \in H^1_0(0,\pi) \},
\end{equation*}
we have the following Poincaré inequality:
\begin{equation*}
\langle g(\theta)\widetilde{L}_\kappa^*g(\theta)\rangle_{M_\kappa}=\langle g'(\theta)^2\rangle_{M_\kappa}+(n-2)\langle\tfrac1{\sin^2\theta}g(\theta)^2\rangle_{M_\kappa}\geqslant(n-2)\langle(g(\theta))^2\rangle_{M_\kappa}.
\end{equation*}
Hence, for~$n\geqslant3$ and~$g\in V$, using Cauchy-Schwarz inequality, we get that
\begin{equation*}
\langle g(\theta)^2\rangle_{M_\kappa}\leqslant\frac1{n-2}\sqrt{\langle g(\theta)^2\rangle_{M_\kappa}\langle(\widetilde{L}_\kappa^*g(\theta))^2\rangle_{M_\kappa}}.
\end{equation*}
Since~$g_\kappa(\theta)=h_\kappa(\cos\theta)\sin\theta$ belongs to~$V$, we get that~$g^N_{\kappa,\varepsilon}(\theta)=R^N_{\kappa,\varepsilon}(\cos\theta)\sin\theta$ also belongs to~$V$.

We are now ready to do the estimations. For~$f$ such that~$\theta\mapsto f(\theta) \sin^{\frac{n}2} \theta$ belongs to~$L^2(0,\pi)$, we get, using Cauchy-Schwarz inequality,
\begin{align*}
|\langle f(\theta)R^N_{\kappa,0}(\cos\theta)&\sin^2\theta\rangle_{M_\kappa}|\leqslant\sqrt{\langle(R^N_{\kappa,0}(\cos\theta))^2\sin^2\theta\rangle_{M_\kappa}\langle f(\theta)^2\sin^2\theta\rangle_{M_\kappa}}\\
&\leqslant\frac1{n-2}\sqrt{\langle(\widetilde{L}_\kappa^*R^N_{\kappa,0}(\cos\theta))^2\sin^2\theta\rangle_{M_\kappa}}\sqrt{\langle f(\theta)^2\sin^2\theta\rangle_{M_\kappa}}\\
&\leqslant\frac1{n-2}{\kappa^{N+1}}\sqrt{\langle(D(H_N)(\cos\theta))^2\sin^2\theta\rangle_{M_\kappa}}\sqrt{\langle f(\theta)^2\sin^2\theta\rangle_{M_\kappa}}.
\end{align*}
Hence using the expansion as~$\kappa\to0$ given in~\eqref{exp_f_zero}, we get the final estimation:
\begin{equation}
\langle f(\theta)R^N_{\kappa,0}(\cos\theta)\sin^2\theta\rangle_{M_\kappa}=O(\kappa^{N+1})\text{ as }\kappa\to0.
\label{estimate_R_zero}
\end{equation}
Similarly, if~$|f(\theta)|=O(\theta^{2\beta})$ in the neighborhood of~$0$, using Lemma~\ref{lem_exp_f_inf}, we get
\begin{align*}
|\langle f(\theta)R^N_{\kappa,\infty}(\cos\theta)&\sin^2\theta\rangle_{M_\kappa}|^2\leqslant\frac1{(n-2)^2}\langle f(\theta)^2\sin^2\theta\rangle_{M_\kappa}\\
&\quad \times\left\langle[L(G^N_N)(\cos\theta)\kappa^{-N}+\sum_{p=0}^{N-1}\kappa^{-p}O(\theta^{2(N-p)})]^2\sin^2\theta\right\rangle_{M_\kappa}\\
&\leqslant O(\kappa^{-2\beta-1})\times O(\kappa^{-2N-1}),
\end{align*}
which gives
\begin{equation*}
\langle f(\theta)R^N_{\kappa,\infty}(\cos\theta)\sin^2\theta\rangle_{M_\kappa}=O(\kappa^{-\beta-N-1})\text{ as }\kappa\to\infty.
\end{equation*}

Now, since we have the expression~\eqref{expression_H_G} of the polynomials~$G^N_p$, we get, by definition of~$R^N_{\kappa,\infty}$,
\begin{align*}
R^N_{\kappa,\infty}(\mu)&=R^{N+1}_{\kappa,\infty}(\mu)+\sum_{p=0}^N(G^{N+1}_p-G^N_p)(\mu)\kappa^{-p}+G^{N+1}_{N+1}(\mu)\kappa^{-N-1}\\
&=R^{N+1}_{\kappa,\infty}(\mu)+\sum_{p=0}^{N+1}a_{N+1-p}^p(1-\mu)^{N+1-p}\kappa^{-p}.
\end{align*}
Since~$(1-\cos\theta)^k=O(\theta^{2k})$, we finally get, using Lemma~\ref{lem_exp_f_inf},
\begin{align*}
\langle f(\theta)R^N_{\kappa,\infty}(\cos\theta)\sin^2\theta\rangle_{M_\kappa}&=\langle f(\theta)R^{N+1}_{\kappa,\infty}(\cos\theta)\sin^2\theta\rangle_{M_\kappa}+\sum_{p=0}^{N+1}\kappa^{-p}O(\kappa^{-\beta-1-N-1+p})\\
&=O(\kappa^{-\beta-N-2})\quad\text{as }\kappa\to\infty.
\end{align*}
This ends the proof of Proposition~\ref{expansion_h}, in the case~$n\geqslant3$.

We suppose now that~$n=2$.
The case~$\kappa\to0$ is easy, since we have the following Poincaré inequality:
\begin{equation*}
\langle g(\theta)\widetilde{L}_\kappa^*g(\theta)\rangle_{M_\kappa}=\langle g'(\theta)^2\rangle_{M_\kappa}\geqslant\frac{e^{-\kappa}\int_0^\pi g'(\theta)^2\mathrm d\theta}{\int_0^\pi e^{\kappa\cos\theta}\mathrm d\theta}\geqslant\frac{e^{-\kappa}\int_0^\pi g(\theta)^2\mathrm d\theta}{\int_0^\pi e^{\kappa\cos\theta}\mathrm d\theta}\geqslant e^{-2\kappa}\langle g(\theta)^2\rangle_{M_\kappa}.
\end{equation*}
We get the same estimations, replacing~$(n-2)$ by~$e^{-2\kappa}$: 
\begin{equation*}
|\langle f(\theta)R^N_{\kappa,0}(\cos\theta)\sin^2\theta\rangle_{M_\kappa}|\leqslant e^{2\kappa}{\kappa^{N+1}}\sqrt{\langle(D(H_N)(\cos\theta))^2\sin^2\theta\rangle_{M_\kappa}\langle f(\theta)^2\sin^2\theta\rangle_{M_\kappa}},
\end{equation*}
which gives the estimate~\eqref{estimate_R_zero} since~$e^{2\kappa}=O(1)$ when~$\kappa\to0$.

The case~$\kappa\to\infty$ is different, since we are not able to get a better Poincaré constant. But we have an explicit expression for~$g_\kappa(\theta)=h(\cos\theta)\sin\theta$, given by~\eqref{def_g2D}:
\begin{equation*}
g_\kappa(\theta)=\frac{\theta}{\kappa}-\frac{\pi}{\kappa}\frac{\int_0^\theta e^{-\kappa \cos\varphi}\mathrm d\varphi}{\int_0^\pi e^{-\kappa \cos\varphi}\mathrm d\varphi}\, .
\end{equation*}
It is also easy to see that, in this case, the coefficients~$a_k^p$ appearing in the definition~\eqref{expression_H_G} of the polynomials~$G_p^N$, are zero when~$p\geqslant1$. Therefore we get that~$G_p^N=0$ for~$p\geqslant1$.

We have~$D(G_1^N)(\cos\theta)=1+O(\theta^{2N})$, so with the formula~\eqref{D_theta}, we obtain
\begin{equation*} 
\partial_\theta(G_1^N(\cos\theta)\sin\theta)=1+O(\theta^{2N}),
\end{equation*}
so we get that~$G_1^N(\cos\theta)\sin\theta=\theta+O(\theta^{2N+1})$ since~$\theta\mapsto G_1^N(\cos\theta)$ is continuous as~$\theta\to0$. Actually, this is the Euler formula for~$\arctan$: if we write~$t=\tan\frac{\theta}2$ we get~$1-\cos\theta=\frac{2t^2}{1+t^2}$,~$\sin\theta=\frac{2t}{1+t^2}$, and then, using the expression~\eqref{expression_H_G} of the polynomials~$G_1^N$ with~$a_k=\frac{2^k(k!)^2}{(2k+1)!}$, we obtain
\begin{equation*}
\arctan t=\frac{t}{1+t^2}\sum_{k=0}^N\frac{2^{2k}(k!)^2}{(2k+1)!}\frac{t^{2k}}{(1+t^2)^k}+O(t^{2N+1}).
\end{equation*}
Now, using the explicit expression of~$g_k$, we have 
\begin{align*}
R^N_{\kappa,\infty}(\cos\theta)\sin\theta&=g_\kappa(\theta)-G_1^N(\cos\theta)\sin\theta\,\kappa^{-1}\\
&=\kappa^{-1}(\theta-G_1^N(\cos\theta)\sin\theta)-\frac{\pi}{\kappa}\frac{\int_0^\theta e^{-\kappa \cos\varphi}\mathrm d\varphi}{\int_0^\pi e^{-\kappa \cos\varphi}\mathrm d\varphi}\\
&=\kappa^{-1}\,O(\theta^{2N+1}) - r_\kappa^\infty(\theta).
\end{align*}
We have~
\begin{equation*}
r_\kappa^\infty(\theta)\leqslant\frac{\pi^2}{\kappa}\frac{e^{-\kappa\cos\theta}}{\int_0^\pi e^{-\kappa \cos\varphi}\mathrm d\varphi},
\end{equation*}
and so, using the estimate~\eqref{int_exp} with~$n=2$, we get
\begin{equation*}
\langle(r_\kappa^\infty(\theta))^2\rangle_{M_\kappa}\leqslant\frac{\pi^4}{\kappa^2}\frac{\int_0^\pi e^{-2\kappa\cos\theta}e^{\kappa\cos\theta}\mathrm d\theta}{(\int_0^\pi e^{-\kappa \cos\varphi}\mathrm d\varphi)^2\int_0^\pi e^{\kappa \cos\theta}\mathrm d\theta}
\leqslant\frac{\pi^4}{\kappa^2(\int_0^\pi e^{\kappa \cos\theta}\mathrm d\theta)^2}=O(\kappa^{-1}e^{-2\kappa}).
\end{equation*}
Therefore, using Cauchy-Schwarz inequality and Lemma~\ref{lem_exp_f_inf},
\begin{equation*}
|\langle f(\theta)r_\kappa^\infty(\theta)\sin\theta\rangle_{M_\kappa}|\leqslant\sqrt{\langle f(\theta)^2\sin^2\theta\rangle_{M_\kappa}\langle(r_\kappa^\infty(\theta))^2\rangle_{M_\kappa}}=O(\kappa^{-\beta-1}e^{-2\kappa}),
\end{equation*}
so we get the final estimate
\begin{equation*}
\langle f(\theta)R^N_{\kappa,\infty}(\cos\theta)\sin^2\theta\rangle_{M_\kappa}=O(\kappa^{-\beta-N-2})\text{ as }\kappa\to\infty,
\end{equation*}
and this ends the proof of Proposition~\ref{expansion_h}.
\subsection{Tips for the general case}
\label{tips}
Here we give some tips to perform an asymptotic study of the coefficients when~$\nu$ depends also on~$\omega\cdot\Omega$.

We will have to take averages against functions of the form~$\theta\mapsto e^{\widehat{\kappa}(\rho,\cos\theta)}$, where
\begin{equation*}
\widehat{\kappa}(\rho,\mu)=\frac1{d(\rho)}\int_0^\mu{\nu(\rho,x)}\mathrm d x.
\end{equation*}
We want to get for example an expansion as the noise~$d$ is large or small. So we are only interested in the dependence on~$\cos\theta$, and we will drop the dependence on~$\rho$ for clarity. We suppose that the function~$\theta\mapsto\nu(\cos\theta)$ is positive, smooth, bounded below and above, and we introduce the parameter~$\kappa=\frac1d$, trying to expand with respect to~$\kappa$.
We write~$\sigma(\mu)=\int_0^\mu\nu(x)\mathrm d x$, so we have~$\widehat{\kappa}(\mu)=\kappa\,\sigma(\mu)$.

The first step consists in the expansion of~$\langle f(\theta)\rangle_{\widehat{M}_\kappa}$ given by
\begin{equation*} 
\langle f(\theta)\rangle_{\widehat{M}_\kappa}= \frac{\int_0^\pi f(\theta) e^{\kappa\,\sigma(\cos\theta)} \sin^{n-2}\theta\, \mathrm d\theta}{\int_0^\pi e^{\kappa\,\sigma(\cos\theta)} \sin^{n-2}\theta\, \mathrm d\theta}.
\end{equation*}

As before, we can easily do a Taylor expansion when~$\kappa\to0$, and we get a result similar to~\eqref{exp_c1_zero} involving quantities of the form~$\int_0^\pi f(\theta)\sigma(cos\theta)^p\mathrm d\theta$. Unless we know explicitly~$\sigma$, we cannot say anything interesting.

When~$\kappa\to\infty$, the strategy is the same: we do the change of variable, setting~$t=\sigma(1)-\sigma(\cos\theta))$, and~$a(t)=\sigma^{-1}(\sigma(1)-t)$, where~$\sigma^{-1}$ is the inverse function of~$\sigma$ (which is increasing since~$\nu>0$, actually we have~$a(t)=\cos\theta$). 
We get: 
\begin{equation*}
 \int_0^\pi f(\theta) e^{\kappa\,\sigma(\cos\theta)} \sin^{n-2}\theta\, \mathrm d\theta=e^{\kappa\,\sigma(1)}\int_0^T\frac{f(\arccos a(t))e^{-\kappa t}}{\nu( a(t))(1- a(t)^2)^{\frac{n-3}2}}dt,
\end{equation*}
where~$T=\sigma(1)-\sigma(-1)$. 
Since~$a(0)=1$, if we know the expansion of~$f$ around~$0$ and~$\nu$ in the neighborhood of~$1$, it only remains to get a Taylor expansion of~$a$ around~$0$ to use Lemma~\ref{watson} (Watson's Lemma).
We can compute the derivatives of~$a$ by induction. 
We have
\begin{equation*}
a'(t)=-\frac1{\nu( a(t))},
\end{equation*}
and this gives immediately~$a^{(n)}=F_n( a(t))$, with the following induction relation for~$F_n$
\begin{equation*}
F_1(\mu)=-\frac1{\nu(\mu)}, \quad F_{n+1}(\mu)=-\frac1{\nu(\mu)}\frac{\mathrm d}{\mathrm d \mu}(F_n(\mu)).
\end{equation*}
This gives us the Taylor expansion of~$a$ at~$t=0$ up to order~$N$:
\begin{equation*}
 a(t)=1+\sum_{n=1}^N f_n(1)\frac{t^n}{n!}+O(t^{N+1}).
\end{equation*}
Since we have~$a'(0)<0$, it is possible to get the analogous of Lemma~\ref{lem_exp_f_inf}: for a function~$f$ such that~$f(\theta)=O(\theta^{2\beta})$, then~$\langle f(\theta)\rangle_{\widehat{M}_\kappa}=O(\kappa^{-\beta})$ as~$\kappa\to\infty$.
Finally, we get an expansion of~$\langle f(\theta)\rangle_{\widehat{M}_\kappa}$ which depend only on the first derivatives of~$\nu$ at~$1$ and on the local behavior of~$f$ around~$0$.

The second step consists in expanding~$\langle f(\theta)\widehat{g}_\kappa(\theta)\sin\theta\rangle_{\widehat{M}_\kappa}$, where~$\widehat{L}^*_\kappa\widehat{g}_\kappa(\theta)=\sin\theta$, the operator~$\widehat{L}^*_\kappa$ being defined by
\begin{equation*}
\widehat L_\kappa^*g(\theta)=-\sin^{2-n}\theta e^{-\kappa \,\sigma(\cos\theta)}\tfrac{\mathrm d}{\mathrm d\theta}(\sin^{n-2}\theta e^{\kappa \,\sigma(\cos\theta)}g'(\theta))+\tfrac{n-2}{\sin^2\theta}g(\theta).
\end{equation*}

It is easy to see that we have
\begin{equation*}
\widetilde L_\kappa^*(P(\cos\theta)\sin\theta)=(L(P)+\kappa\nu(\cos\theta)D(P))(\cos\theta)\sin\theta,
\end{equation*}
so if we set~$\widehat{g}_\kappa(\theta)=\widehat{h}_\kappa(\cos\theta)\sin\theta$, we can decompose~$\widehat{h}_k$ in a way similar to Proposition~\ref{expansion_h}:
\begin{align*}
\widehat h_\kappa(\cos\theta)&=\sum_{p=0}^N \widehat H_p(\cos\theta)\kappa^p+ \widehat R^N_{\kappa,0}(\cos\theta),\\
\widehat h_\kappa(\cos\theta)&=\sum_{p=1}^N \widehat G_p^N(\cos\theta)\kappa^{-p}+ \widehat R^N_{\kappa,\infty}(\cos\theta),
\end{align*}
where~$\widehat H_p$ are the functions (not necessarily polynomials) and~$\widehat G_p^N$ the polynomials of degree at most~$N-p$ given by the following induction relations:
\begin{align*}
&\begin{cases} L(\widehat H_0)=1\\ L(\widehat H_{p+1})(\mu)=-\nu(\mu)D(\widehat H_p)(\mu) \end{cases} \\
&\\
&\begin{cases} D(\widehat G^N_{1})(\cos\theta)=1+O(\theta^{2N}) \\ (D(\widehat G^N_{p+1})+\frac1{\nu(\cos\theta)}L(\widehat G^N_p))(\cos\theta)=O(\theta^{2(N-p)}).\end{cases}
\end{align*}
Again, the remainders satisfy the following estimations, for any function~$f$ such that~$\theta\mapsto f(\theta) \sin^{\frac{n}2} \theta$ belongs to~$L^2(0,\pi)$ and such that~$|f(\theta)|=O(\theta^{2\beta})$ in the neighborhood of~$0$:
\begin{align*}
\langle f(\theta)\widehat R^N_{\kappa,0}(\cos\theta)\sin^2\theta\rangle_{\widehat M_\kappa}&=O(\kappa^{N+1})\text{ as }\kappa\to0, \\
\langle f(\theta)\widehat R^N_{\kappa,\infty}(\cos\theta)\sin^2\theta\rangle_{\widehat M_\kappa}&=O(\kappa^{-\beta-N-2})\text{ as }\kappa\to\infty,
\end{align*}
and this allows to get an expansion of~$\langle f(\theta)\widehat h_\kappa(\cos\theta)\sin^2\theta\rangle_{\widehat M_\kappa}$ when~$\kappa\to0$ and when~$\kappa\to\infty$.


\begin{thebibliography}{10}

\bibitem{agueh2011analysis}
M.~Agueh, R.~Illner, and A.~Richardson.
\newblock Analysis and simulations of a refined flocking and swarming model of
  {Cucker-Smale} type.
\newblock {\em Kinetic and Related Models}, 4(1):1--16, 2011.

\bibitem{aldana2003phase}
M.~Aldana and C.~Huepe.
\newblock Phase transitions in self-driven many-particle systems and related
  non-equilibrium models: A network approach.
\newblock {\em Journal of Statistical Physics}, 112(1-2):135--153, 2003.

\bibitem{aoki1982simulation}
I.~Aoki.
\newblock A simulation study on the schooling mechanism in fish.
\newblock {\em Bulletin of the Japanese Society of Scientific Fisheries},
  48:1081--1088, 1982.

\bibitem{aw2000resurrection}
A.~Aw and M.~Rascle.
\newblock Resurrection of “second order” models of traffic flow.
\newblock {\em SIAM Journal on Applied Mathematics}, 60(3):916--938, 2000.

\bibitem{bazazi2008collective}
S.~Bazazi, J.~Buhl, J.J. Hale, M.L. Anstey, G.A. Sword, S.J. Simpson, and I.D.
  Couzin.
\newblock Collective motion and cannibalism in locust migratory bands.
\newblock {\em Current Biology}, 18(10):735--739, 2008.

\bibitem{bellomo2009mathematics}
N.~Bellomo, H.~Berestycki, F.~Brezzi, and J.-P. Nadal.
\newblock Mathematics and complexity in life and human sciences.
\newblock {\em Mathematical Models and Methods in Applied Sciences},
  19:1385--1389, 2009.

\bibitem{bender1999advanced}
C.~M. Bender and S.~A. Orszag.
\newblock {\em Advanced Mathematical Methods for Scientists and Engineers:
  Asymptotic Methods and Perturbation Theory}.
\newblock Springer-Verlag, New York, 1999.

\bibitem{berthelin2008model}
F.~Berthelin, P.~Degond, M.~Delitala, and M.~Rascle.
\newblock A model for the formation and evolution of traffic jams.
\newblock {\em Archive for Rational Mechanics and Analysis}, 187(2):185--220,
  2008.

\bibitem{bolley2011meanfield}
F.~Bolley, J.~A. Cañizo, and J.~A. Carrillo.
\newblock Mean-field limit for the stochastic {Vicsek} model.
\newblock preprint \verb?arXiv:1102.1325?, 2011.

\bibitem{bolley2011stochastic}
F.~Bolley, J.~A. Cañizo, and J.~A. Carrillo.
\newblock Stochastic mean-field limit: non-{Lipschitz} forces \& swarming.
\newblock To appear in Mathematical Models and Methods in Applied Sciences,
  2011. \verb?arXiv:1010.5405?, 2011.

\bibitem{chate2008collective}
H.~Chaté, F.~Ginelli, G.~Grégoire, and F.~Raynaud.
\newblock Collective motion of self-propelled particles interacting without
  cohesion.
\newblock {\em Physical Review E}, 77(4):046113, 2008.

\bibitem{couzin2002collective}
I.~D. Couzin, J.~Krause, R.~James, G.~D. Ruxton, and N.~R. Franks.
\newblock Collective memory and spatial sorting in animal groups.
\newblock {\em Journal of Theoretical Biology}, 218(1):1--11, 2002.

\bibitem{daganzo1995requiem}
C.~F. Daganzo.
\newblock Requiem for second-order fluid approximations of traffic flow.
\newblock {\em Transportation Research Part B}, 29(4):277--286, 1995.

\bibitem{degond2011macroscopic}
P.~Degond, A.~Frouvelle, and J.-G. Liu.
\newblock Macroscopic limits and phase transition in a system of self-propelled
  particles.
\newblock in preparation.

\bibitem{degond2008continuum}
P.~Degond and S.~Motsch.
\newblock Continuum limit of self-driven particles with orientation
  interaction.
\newblock {\em Mathematical Models and Methods in Applied Sciences},
  18:1193--1215, 2008.

\bibitem{degond2008large}
P.~Degond and S.~Motsch.
\newblock Large scale dynamics of the persistent turning walker model.
\newblock {\em Journal of Statistical Physics}, 131(6):989--1021, 2008.

\bibitem{degond2010macroscopic}
P.~Degond and S.~Motsch.
\newblock A macroscopic model for a system of swarming agents using curvature
  control.
\newblock preprint \verb?arXiv:1010.5405?, 2010.

\bibitem{degond2010diffusion}
P.~Degond and T.~Yang.
\newblock Diffusion in a continuum model of self-propelled particles with
  alignment interaction.
\newblock {\em Mathematical Models and Methods in Applied Sciences},
  20:1459--1490, 2010.

\bibitem{duan2010kinetic}
R.~Duan, M.~Fornasier, and G.~Toscani.
\newblock A kinetic flocking model with diffusion.
\newblock {\em Communications in Mathematical Physics}, 300:95--145, 2010.

\bibitem{fan2004pattern}
H.~Fan and H.~Liu.
\newblock Pattern formation, wave propagation and stability in conservation
  laws with slow diffusion and fast reaction.
\newblock {\em Journal of Hyperbolic Differential Equations}, 1(4):605--626,
  2004.

\bibitem{frouvelle2011dynamics}
A.~Frouvelle and J.-G. Liu.
\newblock Dynamics in a kinetic model of oriented particles with phase
  transition.
\newblock preprint \verb?arXiv:1101.2380?, 2011.

\bibitem{gardiner1985handbook}
C.~W. Gardiner.
\newblock {\em Handbook of Stochastic Methods for Physics, Chemistry and the
  Natural Sciences}, volume~13 of {\em Springer Series in Synergetics}.
\newblock Springer-Verlag, Berlin, second edition, 1985.

\bibitem{gautrais2009analyzing}
J.~Gautrais, C.~Jost, M.~Soria, A.~Campo, S.~Motsch, R.~Fournier, S.~Blanco,
  and G.~Theraulaz.
\newblock Analyzing fish movement as a persistent turning walker.
\newblock {\em Journal of Mathematical Biology}, 58(3):429--445, 2009.

\bibitem{gregoire2004onset}
G.~Grégoire and H.~Chaté.
\newblock Onset of collective and cohesive motion.
\newblock {\em Physical Review Letters}, 92(2):025702, 2004.

\bibitem{ha2008particle}
S.-Y. Ha and E.~Tadmor.
\newblock From particle to kinetic and hydrodynamic descriptions of flocking.
\newblock {\em Kinetic and Related Models}, 1(3):415--435, 2008.

\bibitem{hsu2002stochastic}
E.~P. Hsu.
\newblock {\em Stochastic Analysis on Manifolds}, volume~38 of {\em Graduate
  Series in Mathematics}.
\newblock American Mathematical Society, Providence, Rhode Island, 2002.

\bibitem{huth1992simulation}
A.~Huth and C.~Wissel.
\newblock The simulation of the movement of fish schools.
\newblock {\em Journal of Theoretical Biology}, 156(3):365--385, 1992.

\bibitem{keyfitz1993multiphase}
B.~L. Keyfitz.
\newblock Multiphase saturation equations, change of type and inaccessible
  regions.
\newblock In {\em Flow in porous media: proceedings of the Oberwolfach
  Conference, June 21-27, 1992}, volume 114 of {\em International Series of
  Numerical Mathematics}, pages 103--116, Basel, 1993. Birkhäuser.

\bibitem{lefloch2002hyperbolic}
P.~G. LeFloch.
\newblock {\em Hyperbolic Systems of Conservation Laws: The Theory of Classical
  and Nonclassical Shock Waves}.
\newblock Lectures in Mathematics ETH Zürich. Birkhäuser Verlag, Basel, 2002.

\bibitem{li2008minimal}
Y.-X. Li, R.~Lukeman, and L.~Edelstein-Keshet.
\newblock Minimal mechanisms for school formation in self-propelled particles.
\newblock {\em Physica D}, 237:699--720, 2008.

\bibitem{liu2005axial}
H.~Liu, H.~Zhang, and P.~Zhang.
\newblock Axial symmetry and classification of stationary solutions of
  {Doi-Onsager} equation on the sphere with {Maier-Saupe} potential.
\newblock {\em Communications in Mathematical Sciences}, 3(2):201--218, 2005.

\bibitem{mckean1967propagation}
H.~P. McKean.
\newblock Propagation of chaos for a class of non-linear parabolic equations.
\newblock In {\em Stochastic Differential Equations}, volume~7 of {\em Lecture
  Series in Differential Equations}, pages 41--57. Catholic University,
  Washington, D. C., 1967.

\bibitem{motsch2010numerical}
S.~Motsch and L.~Navoret.
\newblock Numerical simulations of a non-conservative hyperbolic system with
  geometric constraints describing swarming behavior.
\newblock submitted, \verb?arXiv:0910.2951?, 2010.

\bibitem{nagasawa1987propagation}
M.~Nagasawa and H.~Tanaka.
\newblock On the propagation of chaos for diffusion processes with drift
  coefficients not of average form.
\newblock {\em Tokyo Journal of Mathematics}, 10(2):403--418, 1987.

\bibitem{oelschlager1984martingale}
K.~Oelschläger.
\newblock A martingale approach to the law of large numbers for weakly
  interacting stochastic processes.
\newblock {\em The Annals of Probability}, 12(2):458--479, 1984.

\bibitem{reynolds1987flocks}
C.~W. Reynolds.
\newblock Flocks, herds and schools: A distributed behavioral model.
\newblock {\em Computer Graphics}, 21(4):25--34, 1987.

\bibitem{serre1996systemes}
D.~Serre.
\newblock {\em Systèmes de Lois de Conservation {I}: Hyperbolicité,
  Entropies, Ondes de Choc}.
\newblock Fondations. Diderot Éditeur, Paris, 1996.

\bibitem{spohn1991large}
H.~Spohn.
\newblock {\em Large Scale Dynamics of Interacting Particles}.
\newblock Texts and Monographs in Physics. Springer-Verlag, Heidelberg, 1991.

\bibitem{sznitman1991topics}
A.-S. Sznitman.
\newblock Topics in propagation of chaos.
\newblock In {\em École d’Été de Probabilités de Saint-Flour XIX ---
  1989}, volume 1464 of {\em Lecture Notes in Mathematics}, pages 165--251.
  Springer, Berlin, 1991.

\bibitem{vicsek1995novel}
T.~Vicsek, A.~Czirók, E.~Ben-Jacob, I.~Cohen, and O.~Shochet.
\newblock Novel type of phase transition in a system of self-driven particles.
\newblock {\em Physical Review Letters}, 75(6):1226--1229, 1995.

\bibitem{yates2009inherent}
C.A. Yates, R.~Erban, C.~Escudero, I.D. Couzin, J.~Buhl, I.G. Kevrekidis, P.K.
  Maini, and D.J.T. Sumpter.
\newblock Inherent noise can facilitate coherence in collective swarm motion.
\newblock {\em Proceedings of the National Academy of Sciences},
  106(14):5464--5469, 2009.

\end{thebibliography}
\end{document}